\documentclass[12pt]{article}
\usepackage[utf8]{inputenc}
\usepackage{fullpage}
\usepackage{amsfonts}
\usepackage{amsthm}
\usepackage{setspace}
\usepackage{float}
\usepackage{cancel}
\usepackage{algorithm}
\usepackage{algorithmic}
\usepackage[caption = false]{subfig}
\usepackage[final]{graphicx}
\usepackage{xcolor}
\usepackage{mathtools}
\usepackage{physics}
\usepackage{ulem}
\usepackage{verbatim}

\usepackage{diagbox}
\usepackage{lscape}
\usepackage{multirow}
\usepackage{url}

\def\argmax{\mathop{\rm argmax}}

\newtheorem{lemma}{Lemma}
\newtheorem{theorem}{Theorem}
\allowdisplaybreaks

\usepackage[figuresright]{rotating}
\def\boxit#1{\vbox{\hrule\hbox{\vrule\kern6pt\vbox{\kern6pt#1\kern6pt}\kern6pt\vrule}\hrule}}

\title{High dimensional discriminant rules with shrinkage estimators of the covariance matrix and mean vector}
\author{Jaehoan Kim, Hoyoung Park and Junyong Park}
\date{November 2022}

\author{Jaehoan Kim\thanks{Department of Statistics, Texas A$\&$M University, TX, USA, \texttt{k1mjh6561@tamu.edu}},~ 
Junyong Park\thanks{Department of Statistics, Seoul National University, Seoul, Korea,  \texttt{junyongpark@snu.ac.kr}}, and 
Hoyoung Park\thanks{Department of Statistics, Sookmyung Women’s University, Seoul, Korea,  \texttt{hyparks@sookmyung.ac.kr}}
}

\begin{document}

\maketitle
\begin{abstract}
 Linear discriminant analysis (LDA) is a typical method for classification problems with large dimensions and small samples. There are various types of LDA methods that are based on the different types of estimators for the covariance matrices and mean vectors. 
  % Although there are many methods for estimating the inverse matrix of covariance and the mean vectors, 
  In this paper, we consider shrinkage methods based on a non-parametric approach. 
  For the precision matrix, methods based on the sparsity structure or data splitting are examined. Regarding the estimation of mean vectors, Non-parametric Empirical Bayes 
  (NPEB) methods and Non-parametric Maximum Likelihood Estimation (NPMLE) methods, also known as $f$-modeling and $g$-modeling, respectively, are adopted. 
The performance of linear discriminant rules based on combined estimation strategies of the covariance matrix and mean vectors are analyzed in this study.
Particularly, the study presents a theoretical result on the performance of the NPEB method and compares it with previous studies.
Simulation studies with various covariance matrices and mean vector structures are conducted to evaluate the methods discussed in this paper. Furthermore, real data examples such as gene expressions and EEG data are also presented.
   % \textcolor{blue}{
   % Linear discriminant rules in high dimensional situation require precision matrices and mean vectors to be estimated. 
   %In high dimensional situation, several precision estimations are suggested to overcome rank deficiencies, utilizing nonparametric sample splitting scheme or sparsity assumptions. For mean vector estimation, empirical Bayes methods, f-modeling and g-modeling, are widely used and studied. In this paper, we analyzed the performances of linear discriminant rules which are based on combined estimation strategies in both contrived and natural situations. From the result, we obtained reasonable explanations for the discriminant rules which outperform the others under multiple data structures, including the sparsity of the precision matrices and (decorrelated) mean vectors.
%    Supporting this, we provide the sufficient condition on which discriminant rules constructed with the f-modeling method has the asymptotic error rate zero as the number of dimension goes to infinity. Throughout the paper, we suggest the potential comparative advantage of both f-modeling and g-modeling methods.}

\noindent {\bf Keywords} : High dimensional discriminant analysis; Nonparametric maximum likelihood estimation; Nonparametric empirical Bayes; Estimation of precision matrix 
\end{abstract}
%\parkcomment{ Linear discriminant rule or linear discriminant analysis rule?}
%\kimcomment{Resolved: Unified the term as a discriminant rule; Linear discriminant rules, discriminant analysis problem, but linear `discriminant' rule.\\
%Replaced all the `standardization' terms into `decorrelation'. }

\section{Introduction}

Discriminant analysis is one of the most widely emerging problems in various fields, including marketing, biomedical studies, sociology, and psychology. Due to the development of data collection technology, the data which contains the larger number of features than the number of samples became prevalent. Therefore, discriminant analysis in high-dimensional situations became important. Owing to its simplicity and optimality under the knowledge of parameters (see \cite{anderson2003introduction}), the Fisher's linear discriminant analysis (LDA) received a lot of attention and has been studied intensively 
in both theory and practical applications. 
Since this linear discriminant rule requires mean vectors and the precision matrix, they should be estimated in practical situations. 
%\parkcomment{The above sentence may be misleading. 
%The linear discriminant rule is not well defined, so we do not know whether that linear rule needs precision matrix or not.  If the linear rule implies the Fisher's linear discriminant rule, then it makes sense.    }
%\kimcomment{Corrected the term into 'Fisher's linear discriminant rules'.}
However, especially for high dimensional situations, precision matrix should be estimated carefully. Since the sample covariance matrix is not invertible anymore, the inverse of the sample covariance matrix is no more valid as an estimator. Therefore, there have been numerous studies to resolve this problem and construct the linear discriminant rules.

First, there were several approaches to estimating the precision matrix as a diagonal one, by simply assuming the independence among features. Since this assumption makes the true covariance matrix to be diagonal, it only requires each feature’s precision to be estimated. The discriminant rule based on this assumption is called the independence rule (IR). In addition to that, feature selection had been added on this assumption with intuition to “reduce the noise from data”. This led to the feature annealed independence rule (FAIR), studied in \cite{fan2008high}, showing  that the FAIR method can improve the error rate of the IR method in real data situations. 

However, the IR and the FAIR  are  restrictive since they ignore all the possible dependence among features, which might be crucial if there exist some features which are crucially related, for example, voxels in fMRI data. Therefore, there have been several approaches to preserve the `structure', which also can be considered as connectivity, of original data.

One possible approach is to slightly modify the sample covariance matrix into an invertible form using the idea of a shrinkage estimator. For the sample covariance matrix $S$ and $0<\lambda<1$, one can take $S' = (1-\lambda)S + \lambda I$ as the covariance estimator and use its inverse as the estimator of the precision matrix. In addition, other approaches of precision estimation are proposed to solve the optimization problem with penalty terms. See \cite{cai2011constrained} and \cite{friedman2008sparse} for more detailed explanations. \cite{lam2016nonparametric} suggested nonparametric sample splitting technique for the estimation.

%However, to the best of our knowledge, 
%\sout{little is known about these precision estimators when combined with the linear discriminant rules.} 
%these estimators have not been utilized for constructing the linear discriminant rules, and little is known about their performance. 
There is some 
research on utilizing these estimators of covariance or precision matrix in high dimensional discriminant analysis. 
For example, \cite{park2022high} presented intensive comparisons 
among discriminant rules based on 
various estimators of covariance or precision matrix. 
% \sout{It is shown that the estimation of the mean vector based on \btext{NPMLE} and the estimator of the precision matrix in \cite{lam2016nonparametric} provided overall the best performance. }
%\btext
It is shown that the non-parametric mean vector estimation method and the precision estimator proposed in \cite{lam2016nonparametric} provided the best overall performance.
Although \cite{park2022high} presented various results, some important methodologies are missing.  For example, mean vector estimation based on the  NPEB in \cite{greenshtein2009application} has  not been considered 
in their study and the sparse precision matrix has not been combined with such mean vector estimations, either.  

In addition to the estimation of covariance or precision matrix,  
mean vector estimation has several approaches as well. For the general mean estimation approach, \cite{efron2014two}
categorized estimation strategies 
in empirical Bayes into  the $f$-modeling and $g$-modeling
%, which are two different estimation 
which had been actually used in many estimation problems before.   
There are several attempts to utilize these strategies into discriminant analysis problem.  
See  \cite{greenshtein2009application} and 
\cite{efron2009} for $f$-modeling and   \cite{dicker2016High} 
and \cite{park2022high} for $g$-modeling in high dimensional discriminant analysis. 
%Initially unveiled, 
The theoretical properties of 
the discriminant rules based on these estimators 
%enjoy 
are gradually unveiled especially for 
$g$-modelling rather than $f$-modelling; see  \cite{dicker2016High} and \cite{saha2020nonparametric} for theoretical studies 
on the discriminant rule based on $g$-modelling. 
%see \cite{ignatiadis2022confidence} for the confidence interval construction using $g$-modeling.
Although  \cite{greenshtein2009application} 
and \cite{efron2009} 
used the $f$-modelling,  
to the best of our knowledge, 
there has  been no  
study on theoretical properties of 
high dimensional discriminant rule based on $f$-modelling for estimating mean vectors.  
%However, the behavior of these estimators when applied for linear discriminant rules are not well studied. 
%To the best of our knowledge, \cite{park2022high} studied the asymptotic property of the classifiers using $g$-modeling for mean estimation. 
It is worth noting that both $f$-modelling and $g$-modeeling in  mean vector estimations assume independence among features, which requires appropriate data decorrelation before applying these strategies.

In this context, our paper displays two main extensions. First, we elaborate on the impact of the aforementioned estimation methods for precision matrices and mean vectors on the constructed discriminant rules. Since mean vector estimation methods require the precision matrix to remove the dependency in the data before applied, the effect of the two estimators is intertwined. Hence, we cross check the estimation method, we analyze the impact of the estimators on real data examples. Second, we newly analyze the theoretical property of the discriminant rule based on the f-modeling mean vector estimator. With the knowledge about the precision matrix, we consider the asymptotic situation in which the number of features diverges to infinity and compare this rule's asymptotic properties with others. 

This paper is organized as follows. 
In section \ref{sec:Methods}, we review the Fisher's rule and estimation of the covariance matrix and mean vector. 
In section \ref{sec: Asymp}, we provide the asymptotic results for nonparametric empirical Bayes method 
and comparison with other methods. 
Section \ref{sec:Simulations}  
includes simulations studies to evaluate all the discriminant rules considered in this paper for various combinations of mean vectors and covariance matrices.   
In section \ref{sec:Realdata}, real data examples are also presented to compare all the discriminant rules. 
We summarize all the results in section \ref{sec:Concluding} as concluding remarks.

%In estimating the mean vector, using the sample mean vector as it is is the most naive method  used in the Naive Bayes rule. On the other hand, as form of shrinkage estimator, there are nonparametric empirical Bayes (Greenshtein and Park (2009)) and  nonparametric maximum likelihood estimation (NPMLE) \btext{???} which are 
%$f$-modeling and $g$-modeling, respectively, according to Efron (2014).  
%\hparkcomment{I understand NPMLE as a methodology used to estimate the prior distribution in g modeling. Therefore, in the above sentence, the comparison between Efron's g modeling and f modeling should be explained by empirical Bayes, not NPMLE.}

%Although Park et al. (2021) used the Lam () method as the precision matrix estimation method and NPMLE was used for the mean vector, 
%the understanding of the performance according to the combination of various 
%estimation of precision matrix and mean vectors is not sufficient in high dimensional classification. 

\section{Construction of discriminant rules}
\label{sec:Methods}
In this section, we  review the methods of estimating  precision matrix and mean vector in high dimensions.
Throughout this paper, we assume 
${\bf X}_{ij} \sim N_p(\boldsymbol{\mu}_i, \Sigma)$ for $i=1,2$ and $j=1,\cdots, n_i$. To avoid confusion, we use bold symbols to indicate vectors in this section.

A well-known linear discriminant rule, Fisher's linear discriminant analysis (LDA), can be written as  
$$ f (\boldsymbol{x}^{new}) = \frac{1}{2}\, [3- {\rm sign} \left \{\delta (\boldsymbol{x}^{new}) \right\}]$$

for the new data $\boldsymbol{x}^{new} \in \mathbb{R}^p$, where 
\begin{equation}\label{eqn:classification rule}
\delta(\boldsymbol{x}^{new})=\left(\boldsymbol{x}^{new}-\frac{\boldsymbol{\mu}_{1}+\boldsymbol{\mu}_{2}}{2}\right)^\top {\Sigma}^{-1}\left(\boldsymbol{\mu}_{1}-\boldsymbol{\mu}_{2}\right).   
\end{equation}
$\delta(\boldsymbol{x}^{new})$ has the value $1$ or $2$ to indicate the group which the new observation is categorized. 
As one can see, Fisher's LDA requests two estimates, for $\boldsymbol {\mu}_i$ and $\Omega =\Sigma^{-1}$.  
%\hparkcomment{we have to define $x^{new}$}
%\kimcomment{resolved: defined $x^{new}$ above}

First, when estimating $\Omega$, difficulty arises in high-dimensional situations since one cannot directly plug in the inverse of the sample covariance matrix due to its rank deficiency. Therefore, several alternatives exist to estimate $\Omega$ directly in high-dimensional situations. In this paper, we utilized the precision estimation methods using 1) $L^1$ penalty of the precision matrix (which is called the graphical Lasso method; we call this method `glasso' in this paper; see \cite{friedman2008sparse}), and, 2) random sample splitting (we call this method `LAM'; see \cite{lam2016nonparametric}). Of course, we can assume the precision matrix to be diagonal and naively estimate each element. We write this method as `IR,' which is the shorthand for `Independence Rule'.

Regarding the estimation of $\boldsymbol{\mu}_i$, 
we use two methods based on two types of  nonparametric empirical Bayes estimations, namely 
$g$-modeling and $f$-modeling introduced in \cite{efron2014two}.
Such mean vector estimations have been used in \cite{greenshtein2009application} and \cite{dicker2016High} under the independent assumption of variables. Obviously, we also consider the sample mean as an estimate of $\boldsymbol{\mu}_i$.

The mainstream of this paper is to examine the performance of different types of  discriminant rules in high-dimensional situations and analyze their performance both practically and theoretically.

As a building block, we define the decorrelated observation ${\bf Z}_{ij}$ as

\begin{eqnarray}
{\bf Z}_{ij} = \Omega^{1/2} {\bf X}_{ij} \sim N(\Omega^{1/2} {\boldsymbol \mu}_i, I_p) 
\equiv N_p( {\boldsymbol \mu}^*_{i}, I_p )
\label{eqn:decorrelated}
\end{eqnarray}

where $I_p$ is $p \times p$ identity matrix. Based on  ${\bf Z}_{ij}$, the Fisher's rule is expressed as follows: 

\begin{eqnarray}
\delta(\boldsymbol{x}^{new})&=&\left(\boldsymbol{x}^{new}-\frac{\boldsymbol{\mu}_{1}+\boldsymbol{\mu}_{2}}{2}\right)^\top \Sigma^{-1}\left(\boldsymbol{\mu}_{1}-\boldsymbol{\mu}_{2}\right) \nonumber\\
   &=& \left\{\Sigma^{-1/2}\left(\boldsymbol{\mu}_{1}-\boldsymbol{\mu}_{2}\right)\right \}^\top  \Sigma^{-1/2}\left(\boldsymbol{x}^{new}-\frac{\boldsymbol{\mu}_{1}+\boldsymbol{\mu}_{2}}{2}\right)^\top  \nonumber\\
   &=& \left(\boldsymbol{\mu}^*_{1}-\boldsymbol{\mu}^*_{2}\right)^\top \left(\boldsymbol{z}^{new}-\frac{\boldsymbol{\mu}^*_{1}+\boldsymbol{\mu}^*_{2}}{2}\right) \label{eqn: Fisher discriminant rule from decorrelated data}\\
   &\equiv & \sum_{i=1}^p a_i z_i^{new} + a_0 \nonumber
\end{eqnarray}
where $\boldsymbol{z}^{new} =  \Omega^{1/2} \boldsymbol{x}^{new} = \Sigma^{-1/2} \boldsymbol{x}^{new}$,
%% removed `=(z_1^{new},\ldots, z_p^{new})^{\top}' since z_i^{new} is used nowhere
%\kimcomment{I was wondering if the definition of $a_i$ and $a_0$ is necessary, since it is not used in the remainder of the paper.}
%\parkcomment{We need to define all notations, so they should be defined. }
%\kimcomment{Thank you for the comment, professor. What I wanted to ask you was if it is okay to eliminate the fourth line of the above eqnarray.}
%\parkcomment{To emphasize a linear rule, I think it is good to include that. }
%\kimcomment{I understood. Thank you!}
 $a_i = \mu_{1i}^* -\mu_{2i}^*$ 
 and $a_0 = - \left(\boldsymbol{\mu}^*_{1}-\boldsymbol{\mu}^*_{2}\right)^\top  \left(\boldsymbol{\mu}^*_{1}+\boldsymbol{\mu}^*_{2}\right)/2 $. $z_i^{new}$ denotes the $i$-th component of $\boldsymbol{z}^{new}$.
 
In practice, of course, $ \Omega =\Sigma^{-1}$ is unknown, so we need to estimate $\hat \Omega$ and decorrelate ${\bf X}_{ij}$ with $\hat \Omega^{1/2}$.
If such estimates are accurate, one can expect that all variables in ${\bf Z}_{ij}$ are nearly independent, hence we can apply the mean estimation methods described later, which assumes the independence of variables.

To summarize, we need two estimates to construct Fisher's discriminant rule: one for $\Omega$ (or $\Sigma$) and the other one for $\boldsymbol{\mu}_i$ (or $\boldsymbol{\mu}_i^*$). We present these two procedures in the following sections. 

%\hparkcomment{In the case of $\Sigma$, $\Omega$, boldface and non-boldface are mixed. Please unify them.}
%\kimcomment{I converted all $\Sigma, \Omega$ into normal characters.}

\subsection{Estimation of precision matrix}
In the case of the precision matrix estimation,  we present two approaches: one is based on some special structural assumption, and the other one is for the general cases without this assumption. For the former case, it is widely assumed that the precision matrix
has sparsity (relatively small number of nonzero components within the matrix) or graphical structure, as \cite{cai2011constrained} and \cite{friedman2008sparse}. However, this structural assumption has limitations to be applied for general multivariate data. On the other hand, the latter case can be generally utilized since it is free of any structural assumption. See, for example, \cite{lam2016nonparametric}. However, the data dimension gives a computational restriction compared to the former case.
%For the estimation of precision matrix in high dimensional data, several algorithms are suggested. 
In this paper, we consider the method based on the graphical model, called glasso in \cite{friedman2008sparse} as an example of the first approach. For the latter approach, we introduce the method based on random sample data splitting, called the LAM method in \cite{lam2016nonparametric}.
  
%The problem derives from the fact that maximum likelihood estimator (MLE) of covariance matrix ($S_n$) is no longer invertible when $p>n$. 
%Therefore, these methods are contrived to handle the problem with penalizing and random sample splitting, respectively. 

First, as previously mentioned, the glasso method is initially devised to estimate the precision matrix under the sparsity assumption. To utilize this assumption, \cite{friedman2008sparse} converted this estimation problem into the optimization problem, maximizing a penalized log-likelihood function with respect to the matrix $\Theta$. The penalty term is given in the $L_1$-norm of the precision matrix, which can favor the sparsity assumption. The objective function in glasso is written as 
\begin{equation}\label{eqn:graphical lasso}
    f(\Theta) = \log \det(\Theta) - tr(S_n\Theta)-\rho||\Theta||_1
\end{equation}
where $S_n$ is sample covariance matrix,  det$(\Theta)$ 
is the determinant of $\Theta$  and $\rho$ is  the regularization parameter. 
The optimal value $\Theta^* = \argmax_{\Theta}f (\Theta)$ is used as the estimator of $\Omega$. 
%\parkcomment{We need to clarify $\Theta$.}

%In this algorithm, computation time largely depends on the value of $\rho$. 
%As $\rho$ goes to 0, costed time increases drastically. The comparison of computation time is studied in (\cite{friedman2008sparse}). 
    
Subsequently, the LAM algorithm uses a random sample splitting idea. The fundamental issue for high dimensional data is that $S_n$ is not invertible, which hinders $S_n^{-1}$ to be used as an estimator for the precision matrix. To detour this problem, \cite{lam2016nonparametric} applied the following strategy to construct an invertible estimator for the covariance matrix. 

First, the true covariance matrix $\Sigma$ can be represented as
\begin{equation}
    \Sigma = PDP^T, D \succeq 0,
\end{equation}
where $P$ is the orthogonal matrix and $D$ diagonal matrix. This can be also written as
\begin{equation}\label{eqn: LAM sigma method}
    \Sigma = PDP^T = Pdiag(P^T\Sigma P)P^T,
\end{equation}
where $diag(A)$ denotes the diagonal matrix with diagonal elements of $A$. Based on this, to construct an invertible estimator $\hat{\Sigma}$ for $\Sigma$, \cite{lam2016nonparametric} first split $n$ observations into two groups of $n_1$ and $n_2 = n-n_1$ observations. Then, they used one group to estimate $P$ and the other to estimate $\Sigma$ for $diag(P^T\Sigma P)$. Since we already have two groups of observations, we can use the first group of observations to estimate $P$, and the second group of $n_2$ observations to estimate $\Sigma$ on the right-hand side of \eqref{eqn: LAM sigma method}. Therefore, when we denote the sample covariance matrix of two groups as $S_{n, 1}$, $S_{n, 2}$, and $\hat{P_1}$ the eigenvectors of $S_{n, 1}$, the estimator is represented as

\begin{equation}
    \hat{\Sigma}_{LAM} = \hat{P}_1diag(\hat{P}_1^TS_{n, 2}\hat{P}_1)\hat{P}_1^T.
\end{equation}

\subsection{Estimation of mean vector}
%\begin{itemize}
%    \item Brief introduction of mean vector estimation (NPEB, NPMLE)
%\end{itemize}
{
Estimation of the mean vector is another important component of Fisher LDA. 
In many cases, the estimation of the covariance matrix has been focused on high-dimensional discriminant analysis, 
and the sample means vector has been commonly used. 
Shrinkage estimator of the mean vector 
in high dimensional classification 
has been seriously considered 
under the assumption that the covariance matrix is a 
diagonal matrix, which is 
a modified form of the naive Bayes rule. 
\cite{greenshtein2009application} and \cite{dicker2016High} are 
representative work 
applying $f$-modeling and $g$-modeling to discriminant analysis 
under the assumption that all variables are independent. 
}

{ 
When we have independent 
random variables  $Y_i \sim N(\mu_i, 1)$ for  $1\leq i \leq p$,  a lot of research has been done on the simultaneous estimation of the mean vector, $\boldsymbol{\mu}=(\mu_1,\ldots, \mu_p)$. A typical example is the James-Stein estimator, but from a Bayesian point of view, it is an estimator under the assumption that each mean value is generated from a normal distribution, so it  has the disadvantage of not reflecting the various structures of the mean values such as bi-modality of mean values 
which may occur for the case of the sparsity of mean values. 
When $p$ is large, 
there are two typical methods in nonparametric empirical Bayes methods  discussed in \cite{efron2014two} : 
 $f$-modeling based on the estimation of marginal density and $g$-modeling based on the estimation of mixing distribution or prior distribution as a non-parametric method. The detail of these two methods are discussed shortly.
In particular, $f$-modeling has a simpler calculation process than $g$-modeling, however 
with the recent development of various algorithms that can be used for $g$-modeling, empirical Bayesian methods based on $g$-modeling are also widely used.
}

{In fact, when estimating the mean vector via these methods, the observed values are assumed to be independent. However, 
in classification, this assumption is not satisfied in general. Although there exist correlations among all variables,  existing studies such as \cite{greenshtein2009application} applied $f$-modeling method to correlated variables in high dimensional classification problems. 
%\btext
{Instead, since we can obtain an estimate of the precision matrix following the previous section, we can remove the correlation among the variables using the square root matrix of this obtained estimate. } }

%Empirical Bayes method is suggested to estimate mean vector instead of sample mean, which is an inadmissable estimator with respect to mean squared error, according to Stein's paradox. 
More specifically, consider $\boldsymbol{X}_i \sim N_p(\boldsymbol{\mu}, \Sigma)$ with $i = 1, \cdots, n$ and $p \times p$ positive definite matrix $\Omega^{1/2}$ where $\Omega = \Sigma^{-1}$. We consider the random variables $\boldsymbol{Z}_i$ as follows: 

\begin{equation}
\label{eqn:decorrelation}
    \boldsymbol Z_{i} \equiv \Omega^{1/2}\boldsymbol X_{i} \sim N_p(\Omega^{1/2}\boldsymbol \mu, I_p) \equiv N_p(\boldsymbol \mu^*, I_p),
\end{equation}
%\hparkcomment{
%The use of matrix $R$ seems unnecessary. Why don't you just denote $ \Omega^{1/2}$ instead of $R$?
%}
%\kimcomment{I removed $R$s and replaced them into $\Omega^{1/2}$. Thank you for the comment. In addition, I converted all $I$ into $I_p$. Finally, please check page 15 and 36 and remove `btext' if it seems okay to be added. Thank you!}
where $\boldsymbol \mu = (\mu_{1}, \cdots, \mu_{p})^T$, $\boldsymbol \mu^* = (\mu^*_{1}, \cdots, \mu^*_{p})^T$, $\boldsymbol X_i = (X_{i1}, \cdots, X_{ip})^T$ and ${\boldsymbol Z}_i = (Z_{i1}, \cdots, Z_{ip})^T$. 
Therefore, the $j$-th component of $\boldsymbol{Z}_i$, $Z_{ij}$, satisfies  $Z_{ij} \sim N(\mu_{j}^*, 1)$ and is independently distributed. 
{We estimate the mean vector  $\boldsymbol \mu^*_i$ based on decorrelated observations  ${\boldsymbol Z}_i = (Z_{i1}, \cdots, Z_{ip})^T$ and then transform back to $  \hat {\boldsymbol \mu}_i = R^{-1} \hat {\boldsymbol \mu}^*_i$, aligned with the assumption of $f$-modeling and $g$-modeling.} 

%To apply empirical Bayesian structure, distribution about $\boldsymbol\mu$ is assumed as $\mu_{ij}^* \sim G_j$, where $j = 1, \cdots, p$. To estimate $\boldsymbol \mu_i$, $\boldsymbol \mu_i^*$ is estimated under Bayesian scheme and restored as $\boldsymbol\mu_i = R^{-1}\boldsymbol\mu_i^*$. 

%To estimate $\mu_i$, two different methods, which are classified f-modeling and g-modeling respectively, is applied. The first strategy is nonparametric empirical Bayes method (NPEB), and the other one nonparametric maximum likelihood estimator (NPMLE).

Now we provide $f$- and $g$-modeling based estimation methods in detail. We consider one dimensional random variable $Z$ and its observed value $z_1, \cdots, z_n$, under the hierarchical structure  
\begin{equation}
    Z_i \sim N(\mu_i, 1),\quad \mu_i \sim G.
\end{equation}
Here, $G$ is a cumulative distribution function of $\mu$. Under the Bayesian scheme, the probability distribution function of $Z$, $g^*$, is written as
\begin{equation}
    g^*(z) = \int \phi(z-v)\, dG(v),
\end{equation}
where $\phi(\cdot)$ is the probability distribution function of standard normal distribution, say $\phi(x) = \exp(-x^2/2)/\sqrt{2\pi}.$ Under this structure, the Bayes estimator for $\mu_i$ is 

\begin{eqnarray}
 {E(\mu_i |z_i)} 
 =  z_i + \frac{(g^*)'(z_i) }{ g^*(z_i) },
 \label{eqn:GMLE}
\end{eqnarray}
where $(g^*)'(z) = dg^*(z) / dz$. 

There are two categories of estimator for $E(\mu_i | z_i)$: 
First one is $\widehat {g^*(z)} =  \int \phi(z-v) d\hat G(v)$ with an estimator for $G$, called $g$-modeling. See \cite{jiang2009general} and \cite{efron2014two}. 
The other one is estimating $g^*$ directly by using density estimation based on the observed values $z_1,\ldots, z_p$, called $f$-modeling. 
See \cite{greenshtein2009application}. We denote the first method as `NPMLE' since we would estimate $G$ with maximum likelihood estimator (MLE), and the latter one as `NPEB'.

To begin with, the NPMLE method requires $G$ to be estimated and MLE can be used as the solution of the optimization problem 
\begin{equation}
    \hat{G} = \argmax_{F \in \mathcal{F}} {\sum_{i=1}^{n}} \log \{\int \phi(z_i - \mu) dF(\mu) \},\label{eqn: g modeling MLE}
\end{equation}
where $\mathcal{F}$ is a set of distribution functions. It cannot be solved in its original form since $\mathcal{F}$ is infinite-dimensional. Therefore, we must constrain this as a finite-dimensional one. \cite{koenker2014convex} restricted $\mathcal{F}$ into the set of piecewise constant distribution functions with $K+1$ regular grid points and converted this problem into $K$ dimensional convex optimization problem. See \cite{dicker2016High} for details. The well-performing behavior of $\hat{G}$ via this algorithm, even with a relatively small number of $K$ ($K \approx \sqrt{n}$), is justified in \cite{dicker2016High}. This algorithm can be implemented in R via
\texttt{REBayes} package.

In our data setting \eqref{eqn:decorrelation}, 
we use the notation of $z_{ij}$ to denote the $j$-th element of the $\boldsymbol{z}_i$ vector. Let $\bar{z}_j = \displaystyle\sum_{k=1}^{n} z_{kj}  / n_g$, the $j$-th component of the sample mean vector $\bar{\boldsymbol z}$ . Note that $\sqrt{n}\bar{z}_j \sim N(\sqrt{n}\mu_{j}^*, 1)$. Therefore, we substitute $\sqrt{n}\bar{z}_i$ for $z_i$ in \eqref{eqn: g modeling MLE} to eventually obtain the estimate for $\sqrt{n}\boldsymbol\mu_j^*$ and divide it by $\sqrt{n}$.

Subsequently, $f$-modeling method estimates $(g^*)'(z_i)$ and $g^*(z_i)$ of \eqref{eqn:GMLE} using the observed data $z_i$, $i = 1, \cdots, n$. This paper uses kernel estimators with normal kernels. See section 2 of \cite{brown2009nonparametric} for detail. We present the eventual form of the estimator $\hat{\mu}_{EB}$ as follows.

%\kimcomment{revised the following expression into more readable form.}
\begin{equation}\label{eqn: NPEB expression}
    (\hat{\boldsymbol\mu}_{EB})_i = \widehat{E(\boldsymbol\mu_i| \boldsymbol{\bar{z}} )}=   
    \bar{z}_i +\cfrac{\displaystyle\sum_{j=1}^p \left(\bar{z}_j - \bar{z}_i \right)\phi\left \{ \sqrt{n}(\bar{z}_i - \bar{z}_i)/h\right \}}
        {h^2\displaystyle\sum_{j=1}^p \phi\left \{ \sqrt{n}(\bar{z}_i - \bar{z}_i)/h\right \}}.
\end{equation}

Here, $(\hat{\mu}_{EB})_i$ denotes the $i$-th component of $\hat{\boldsymbol\mu}_{EB}$, and $h$ the bandwidth of the kernel estimator. Throughout our paper, we use $h = 1/\sqrt{\log p}$, with which good theoretical properties are known to hold (see \cite{brown2009nonparametric}). It is worth noting that both estimation methods provide the estimator of $\boldsymbol\mu^*$ solely from the sample mean $\bar{\boldsymbol z}$.

\subsection{Construction of discriminant rules}\label{sec: construction of discriminant rules}

First, recall that when the prior probability of group 1 and group 2 is provided, the Bayes discriminant rule $\delta$ has the form of 
\begin{equation}
    \delta(\boldsymbol{x}^{new}) = \left(\boldsymbol{x}^{new}-\frac{\boldsymbol{\mu}_{1}+\boldsymbol{\mu}_{2}}{2}\right)^\top \boldsymbol{\Sigma}^{-1}\left(\boldsymbol{\mu}_{2}-\boldsymbol{\mu}_{1}\right) - \log\cfrac{p_1}{p_2},
\end{equation}
where $p_1, p_2$ denotes the prior probability of each group. The only difference from \eqref{eqn:classification rule} is the last term, which denotes the prior odds.
In real data situations, when $n_1, n_2$ samples are obtained from group 1 and group 2, respectively, prior odds can be %replaced to
estimated as $n_1/n_2$. As discussed in \eqref{eqn: Fisher discriminant rule from decorrelated data}, we can write

\begin{equation}\label{eqn: classification rule oracle ver}
    \delta(\boldsymbol{x}^{new}) = \left(\boldsymbol{\mu}^*_{1}-\boldsymbol{\mu}^*_{2}\right)^\top \left(\boldsymbol{z}^{new}-\frac{\boldsymbol{\mu}^*_{1}+\boldsymbol{\mu}^*_{2}}{2}\right) - \log \frac{p_1}{p_2}.
\end{equation}

As we previously mentioned, we first estimate the precision matrix, decorrelate the data (including $\boldsymbol{x}_{new}$), and then estimate the decorrelated mean vector, $\mu^*$. 
For the classification, we would have two groups of datasets. 
In each group, we obtain the estimator for the precision matrix ($\hat{\Omega}_1$, $\hat{\Omega}_2$). Then, we use the pooled estimator as a final precision estimator. Namely,

$$\hat{\Sigma}^{-1} = \hat{\Omega} = \frac{(n_1-1)\hat{\Omega}_1 + (n_2-1)\hat{\Omega}_2}{n_1+n_2-2}.$$

For $\boldsymbol{X}_i \sim N_p(\boldsymbol{\mu}, \Omega^{-1})$, one can write

 $$
\hat{\Omega}^{1/2}\boldsymbol{X}_i \sim N_p(\hat{\Omega}^{1/2}\boldsymbol{\mu}, \hat{\Omega}^{1/2}\Omega^{-1}\hat{\Omega}^{1/2}), \;i=1, \cdots, n.
 $$ 
 
Assuming that $\hat{\Omega}$ estimates $\Omega$ properly, $\hat{\Omega}^{1/2}\Omega^{-1}\hat{\Omega}^{1/2}$ can be assumed to be similar to $I_p$. Therefore, we can apply NPEB and NPMLE methods under desirable settings, to finally get the estimator of $\hat{\Omega}^{1/2}\boldsymbol\mu$. Since we have two groups of data, we denote the decorrelated mean of group $i$ and its estimate as $\boldsymbol{\mu_i}^*$ and $\widehat{\boldsymbol{\mu_i}^*}$, respectively ($i = 1, 2$).

Following the procedure, one can obtain the practical version of the discriminant rule as

\begin{equation}\label{eqn:discriminant rule - mean each}
    \delta(\boldsymbol{x}^{new}) = \left( \widehat{\boldsymbol{\mu}^*_{1}}-\widehat{\boldsymbol{\mu}^*_{2}}\right)^\top \left( \hat{\Omega}^{1/2}\boldsymbol{x}^{new}-\frac{\widehat{\boldsymbol{\mu}^*_{1}}+\widehat{\boldsymbol{\mu}^*_{2}}}{2}\right) - \log \frac{p_1}{p_2}.
\end{equation}

Instead of estimating ${\boldsymbol\mu}_1^*$ and ${\boldsymbol\mu}_2^*$ separately, we can estimate $\boldsymbol{\mu}_1^* - \boldsymbol{\mu}_2^*$ at once. When $\boldsymbol{z}_1$ and $\boldsymbol{z}_2$ denotes the decorrelated sample mean (with true precision matrix) of the group $1$ and $2$, 

\begin{equation*}
    \boldsymbol{z}_1 \sim N_p({\boldsymbol\mu}_1^*, {n_1}^{-1}I_p), \;%\frac{1}{n_1}I), \;
    \boldsymbol{z}_2 \sim N_p({\boldsymbol\mu}_2^*, {n_2}^{-1}I_p).%\frac{1}{n_2}I).
\end{equation*}
Since $\boldsymbol{z}_1$, $\boldsymbol{z}_2$ are independent,
\begin{equation*}
    \boldsymbol{z}_1 - \boldsymbol{z}_2 \sim N_p\left({\boldsymbol\mu}_1^* - {\boldsymbol\mu}_2^*, \left(
    n_1^{-1} + n_2^{-1} \right)I_p\right).
\end{equation*}

Therefore, $\boldsymbol{\mu}_1 - \boldsymbol{\mu}_2 $ can be directly estimated from $\boldsymbol{z}_1 - \boldsymbol{z}_2 $ using NPEB and NPMLE method by multiplying and dividing the constant $a_n = (1/n_1 + 1/n_2)^{-1/2}$. Since the mean difference is considered to be the main key to separate two groups, we can simply plug in $\boldsymbol{z}_1, \boldsymbol{z}_2$ for $\boldsymbol{\mu}_1, \boldsymbol{\mu}_2$ in $\left(\boldsymbol{\mu}_1^* + \boldsymbol{\mu}_2^*\right)/2$ of \eqref{eqn: classification rule oracle ver}. Therefore, the classifier built on this method is written as

\begin{equation}\label{eqn:discriminant rule - mean diff}
    \delta(\boldsymbol{x}^{new}) = \left( \widehat{\boldsymbol{\mu}^*_{1}- \boldsymbol{\mu}^*_{2}}\right)^\top \left( \hat{\Omega}^{1/2}\boldsymbol{x}^{new}-\frac{\boldsymbol{z}^*_{1}+\boldsymbol{z}^*_{2}}{2}\right) - \log \frac{p_1}{p_2}.
\end{equation}

In section \ref{sec:Realdata}, we compare (\ref{eqn:discriminant rule - mean each}) and (\ref{eqn:discriminant rule - mean diff})'s performance in real data. We noticed that two discriminant rules show similar performance when same estimation strategies are used. Note that (\ref{eqn:discriminant rule - mean each}) and (\ref{eqn:discriminant rule - mean diff}) become identical
when one use the estimator of $\boldsymbol\mu_1^*$, $\boldsymbol\mu_2^*$, and $\boldsymbol\mu_1^* - \boldsymbol\mu_2^*$ based on sample means.
Therefore, we use the discriminant rule (\ref{eqn:discriminant rule - mean diff}) for the simulation and theoretical analysis in section \ref{sec: Asymp}. We denote `NPEB1', `NPMLE1' for the discriminant rule using NPEB, NPMLE method in (\ref{eqn:discriminant rule - mean diff}) and `NPEB2', `NPMLE2' for (\ref{eqn:discriminant rule - mean each}).

%\hparkcomment{Refer to Jiang Zhang (2009) for NPMLE method mean vector estimation. }

%Nonparametric maximum likelihood estimator (NPMLE) is the method suggested in (references: Jiang and Zhang, 2009). According to (reference: Efron, 2014) this method is g-modeling method, which utilizes prior distribution assumption about $G$
    
\section{Asymptotic results of the NPEB method }\label{sec: Asymp}
In this section, we present the asymptotic result 
for the discriminant rule with the NPEB method. 
\cite{wang2018dimension} studied the performance of the discriminant rule (\ref{eqn:discriminant rule - mean diff}) when using sample mean.
\cite{park2022high} studied the performance of (\ref{eqn:discriminant rule - mean diff}) using the mean vector estimator with NPMLE in \eqref{eqn:GMLE}  
and showed a 
region for the parameters 
related to the strength of mean values and 
the level of sparsity such that 
the error rate is asymptotically 0 as the dimension increases. 
We provide an analogous result for the NPEB method 
under the setting in \cite{park2022high} 
and compare the result with those for NPMLE, hard threshold, and naive Bayes rule %\btext
{(which uses the sample mean for \eqref{eqn:discriminant rule - mean diff}).} 

As in \cite{park2022high} and \cite{wang2018dimension}, we assume that the precision matrix is consistently estimated so that 
the data are decorrelated in an appropriate way, 
for example, glasso method provides 
the uniformly consistent estimators of all components in the precision matrix. 
We emphasize the effects of
mean vector estimations after data are decorrelated.   

Under (\ref{eqn:discriminant rule - mean diff}), assuming  $\Sigma = I_p$ or the precision matrix is consistently estimated,  
 the misclassification error rate of a discriminant rule $\delta$ 
 is written as 
\begin{eqnarray*} 
    P\bigg(\delta(x^{new})>0 \,\bigg|\, x^{new} \in \text{group 2} \bigg) &=&
    P\bigg( (z^{new}-\frac{\bar{z}_1+\bar{z}_2}{2})^T \hat\mu_D^* > 0 \, \bigg | \, z^{new} \sim N_p(\mu_2^*, I) \bigg)\\
    &=& \Phi \left\{-\frac{1}{2}\left(\frac{\mu_D^{*T}\hat\mu_D^*}{\lVert\hat\mu_D^*\rVert_2} + \frac{(\bar{z}_1+\bar{z}_2-(\mu_1^*+\mu_2^*))^T\hat\mu_D^*}{\lVert\hat\mu_D^*\rVert_2} \right) \right\},
\end{eqnarray*}
where $\mu_D^* = \mu_1^*-\mu_2^*$, $\hat\mu_D^* = \widehat{\mu_1^*-\mu_2^*}$. 
%Note that under this assumption, $\mu_i$ and $\mu_i^*$ become identical.
Since 
$$\bar{z}_1+\bar{z}_2-(\mu_1^*+\mu_2^*) \sim N\left(0, a_n^{-2} I_p\right),$$ when $n_1$ and $n_2$ are fixed, $\bar{z}_1+\bar{z}_2-(\mu_1^*+\mu_2^*) = O_p(1)$ holds. 
Therefore, from the Cauchy-Schwartz inequality,
we have $$\cfrac{\left\{\bar{z}_1+\bar{z}_2-(\mu_1^*+\mu_2^*)\right\}^T\hat\mu_D^*}{\lVert\hat\mu_D^*\rVert_2} = O_p(1).$$ Thus, as $p \to \infty$, 
we have the result that  
$P\bigg(\delta(x^{new})>0 \, \bigg| \, x^{new} \in \text{group 2} \bigg) \to 0$ is equivalent to 
\begin{equation}\label{eqn:perfect condition}
    V := \frac{\mu_D^{*T}\hat\mu_D^*}{\lVert\hat\mu_D^*\rVert_2} \overset{\mathrm{p}}{\to} \infty. 
\end{equation}

This implies that if $V \overset{\mathrm{p}}{\to} \infty$ as $p \to \infty$, a given decision rule $\delta$ has the error rate $0$ asymptotically. However, as $p/n \to \infty$, it is widely known that $\delta$ has the error rate $1/2$ asymptotically, which matches the intuition (see \cite{bickel2004some}). Therefore, we consider the situation in which both the number of meaningful components (say, the number of `signals') and the degree of signal increase in a positive power of $p$.

%\textcolor{red}
{To implement 
different levels of sparsity 
of signals in high dimension as $p\rightarrow \infty$, 
we parametrize   
the number of nonzero mean differences and their  sizes  
which are commonly used in 
the study of sparse signals in high dimensional data analysis. %\sout{For example, see} %\btext
{See 
\cite{park2022high}  
and \cite{dicker2016High} for the examples.}%\sout{in high dimensional classification.}    
}
More specifically, we set $\mu_1^* = (0, \cdots, 0)$ and  $\mu_2^* = (\mathbf{\Delta}_l^T, \mathbf{0}_{p-l}^T)$ 
where   $\mathbf{\Delta}_l$ is the $l$ dimensional column vector with 
all components $\Delta$. Then we assign
\begin{equation}\label{eqn: assumption for theoretical analysis}
    \Delta = p^b, l = [p^a] 
\end{equation}
for $b>0$ and $0<a<1$. 
$\Delta$  and $l$ represent 
the strength of the signal for each component and 
the level of sparsity, respectively.
%\textcolor{red}
{As $p$ increases, the sparsity level $l=[p^a]$ and the strength of signal $\Delta$ are represented as functions of $p$. This setting was used in the theoretical study by 
\cite{park2022high} and the simulation experiments by \cite{dicker2016High}. This setting will include various forms, including cases where $\Delta$ converges to infinity or to zero, in other words, the strength of the signal $\Delta$ is very large or very small.  This will be accompanied by the number of variables $l$ resulting in comparisons of various methodologies for various cases.}

%\textcolor{red}
{We first set the range of interest for the constants a and b since if either $a$ or $b$ is too large, the two groups can be easily separated regardless of the method used for classification. Therefore, as mentioned in \cite{park2022high}, we want to compare various methods only in the $(a,b)$ range as follows: 
\begin{eqnarray*}
 (a,b) \in {\cal R} \equiv  (T_1 \cup T_2 )^c
 \label{eqn:T1T2}
\end{eqnarray*}
where 
$T_1=\left\{(a, b) \in \mathbb{R}^2, a+2 b \geq 1 / 2, 0<a \leq 0.5\right\}$ and $T_2=\left\{(a, b) \in \mathbb{R}^2: 0.5<a \leq 1\right.$, $\left.b>0\right\}$.
In $T_1$ or $T_2$, the combination of 
$\Delta$ and $l$ provides 
large strengths of signals leading to clear 
separation of two groups 
by various methods. 
See more details in \cite{park2022high}.  
We consider  the outside of $T_1$ and $T_2$ which is 
$(T_1\cup T_2)^c$ 
in comparison of various methods.  
}

Note that if a discriminant rule has the property of \eqref{eqn:perfect condition} in a wider range of $(a,b)$ than another one, then the first discriminant rule has superiority to the other one %\btext
{asymptotically}.    
\cite{park2022high} derived some conditions %\btext
{on $(a, b)$}
under which the classifier (\ref{eqn:discriminant rule - mean diff}) using the NPMLE, the sample mean, and hard threshold method achieves \eqref{eqn:perfect condition}. \cite{park2022high} showed that the discriminant rule with NPMLE has an advantage over 
naive Bayes rule and hard threshold in that 
\eqref{eqn:perfect condition} of NPMLE method is obtained 
in a wider range than the other two methods.

%\sout{We derive the range of $(a,b)$ of NPEB and present the following theorem to compare the results of NPMLE, naive Bayes rule and hard threshold method. }

Here, we suggest the range of parameters to achieve (\ref{eqn:perfect condition}) in the NPEB method and compare this with the regions of NPMLE, the sample mean(SM), and hard threshold(Hard) methodologies in the following theorem. We first define $\mathcal{R}_{NPEB}, \mathcal{R}_{NPMLE}, \mathcal{R}_{SM}$ and $\mathcal{R}_{Hard}$ as the areas of $(a, b)$ satisfying (\ref{eqn:perfect condition}) under the corresponding mean estimation methods. 

%\hparkcomment{
%The naive Bayes rule is a method of estimating average parameters as sample averages under the independent rule. When defining the convergence region below, it seems inappropriate to refer to the mean parameter estimation method based on the Naive Bayes rule because the mean parameters are just estimated as the sample means.
%}
%\kimcomment{I understood. Thank you so much for the explanation!}
%The range of $(a, b)$ in which the discriminant rule satisfies (\ref{eqn:perfect condition}) varies according to the mean estimation method. 

%In addition, in \cite{park2022high}, the method of hard threshold is introduced, and the corresponding area is analyzed under the same setting. We denote this area as $\mathcal{R}_{Hard}$.

\begin{figure}[ht]
    \centering
    \includegraphics[scale = 0.9]{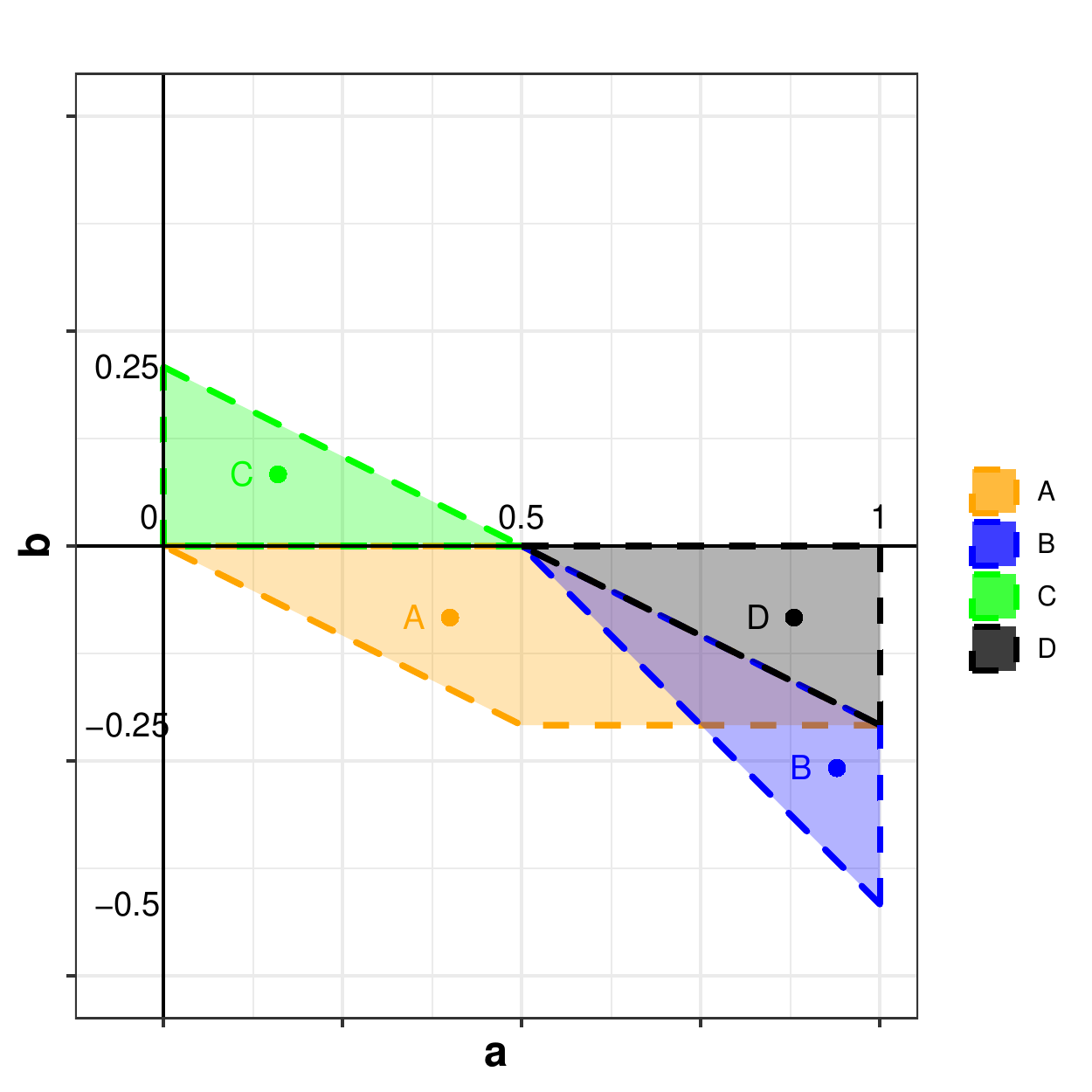}
    \caption{The area of $(a, b)$ under the setting (\ref{eqn: assumption for theoretical analysis})}
    \label{fig:ConvArea_ab}
\end{figure}

\begin{theorem}

Suppose $\mu_1 = (0, \cdots, 0)$ and  $\mu_2 = (\mathbf{\Delta}_l^T, \mathbf{0}_{p-l}^T)$ 
where  $\Delta = p^b, l = [p^a] $. Then, $\mathcal{R}_{NPEB}, \mathcal{R}_{NPMLE}$ cover both $\mathcal{R}_{SM}$ and $\mathcal{R}_{Hard}$. Specifically, 
\begin{align*}
    \mathcal{R}_{Hard} = C,& \, \mathcal{R}_{SM} = D,\, \\
    \mathcal{R}_{NPEB} \supset A \cup C \cup D,& \, \mathcal{R}_{NPMLE} \supset B \cup C \cup D,   
\end{align*}

where

\begin{align*}
    A &= \left\{(a, b)\in \mathbb{R}^2,\, 0 \le a + 2b \le 1/2: -0.25 < b < 0 \textmd{ and } 0<a<1 \right\},   \\
    B &= \left\{(a, b)\in \mathbb{R}^2,\, a + b \ge 1/2: -0.5<b<0 \textmd{ and } 0.5<a<1  \right\},    \\
    C &= \left\{(a, b)\in \mathbb{R}^2,\, a + 2b \le 1/2: 0<b<0.25 \textmd{ and } 0<a<0.5  \right\}, \\
    D &= \left\{(a, b)\in \mathbb{R}^2,\, a + 2b \ge 1/2: -0.25<b<0 \textmd{ and } 0.5<a<1  \right\}.
\end{align*}

\label{thm:NPEB}
\end{theorem}

\begin{proof}
See Appendix.
\end{proof}

%\begin{remark}
%When $(a, b)$ satisfies $a + 2b > 1/4$ and $b>0$, it is known that  (\ref{eqn:perfect condition}) holds for NPMLE, naive Bayes rule and Hard threshold method (see \cite{park2022high}). We mention that the NPEB method guarantees  (\ref{eqn:perfect condition}) as well under this region.
%\end{remark}

%Therefore, other than this common region, we consider the area in which the classification becomes difficult.

Note that $\mathcal{R}_{NPEB}$ and $\mathcal{R}_{NPMLE}$ can be enlarged further, while $\mathcal{R}_{SM}$ and $\mathcal{R}_{Hard}$ exactly matches with $D$ and $C$. The visual illustration of each region is in figure \ref{fig:ConvArea_ab}. We can see that the NPEB method covers all the areas of the naive Bayes rule and hard threshold method. 
On the other hand, neither $\mathcal{R}_{NPEB}$ nor $\mathcal{R}_{NPMLE}$
is included in each other. Both methods are expected to work well 
in both sparse and dense cases while the hard threshold method and naive Bayes rule 
are designed for only sparse and dense cases, respectively. 

We provide numerical studies in section \ref{sec:Simulations} and real data analysis in \ref{sec:Realdata} to compare these methods, 
and see that there doesn't seem to be any tendency that 
either of NPMLE and NPEB dominates the other method.

\section{Simulations}
\label{sec:Simulations}
%According to covariance structure, glasso method and LAM method have their
%own superiority based on motivation. 
%Glasso method is developed for estimating 
%sparse precision matrix in which most of the elements 
%are considered to be zero (see Friedman et al, 2008; Banerjee et al, 2008). 
%In glasso estimation method, adjusting the parameter $\rho$ is important. 
%With relatively small $\rho$ value, glasso algorithm estimates zero elements to be nonzero, since the effect of increased log likelihood exceeds the penalty assigned by its value. In contrast, relatively large $\rho$ value, glasso algorithm ignores most nonzero terms to be zero since its penalty dominates the amount of increased log likelihood. Throughout the simulation, we fixed the $\rho$ value as $\rho = 0.01$. 
%On the other hand, for LAM method, 
%we fixed the split ratio of samples. we used $80\%$ of data to estimate eigenvectors and $20\%$ for eigenvalues. 
In this section, we provide simulation results for 
various combinations of mean vectors and covariance matrices.  
We divided simulation settings according to the sparseness and  denseness  of the precision matrix and  those  of the difference of mean vectors. 
%We first estimate the precision matrix, decorrelate the difference of mean vector ($\Sigma^{-1/2}(\bar{x}_2-\bar{x}_1)$) and apply mean estimation strategies (NPEB, NPMLE) to ultimately estimate $\Sigma^{-1/2}(\mu_2-\mu_1)$. 
%Therefore, the effect of precision estimation and mean vector estimation is inevitably combined. Nevertheless, we scrutinized all combinations to find out the effectiveness of each strategy.

%\subsection{Simulation settings}
Under the assumption of multivariate normal distribution, $\Omega_{ij} = 0$ means that $X_i$ and $X_j$ are conditionally independent. To verify the strength of glasso and LAM methods separately, we divided simulation settings according to the sparseness of the precision matrix or covariance matrix. 
%In appendix, we organized the theoretical properties of mean vector estimation methods according to the sparsity of mean vector difference. 
To compare the mean vector estimation method individually, we assumed the simulation setting as follows. We set $\mu^T_1 = (0, \cdots, 0)$, $\mu^T_2 = (\mathbf{\Delta}_l^T, 0_{p-l}^T)$ 
where  $\mathbf{\Delta}_l$ 
is the $l$ dimensional column vector with all component $\Delta$ 
and $\mathbf{0}_{p-l}$ is the $(p-l)$ dimensional vector with all $0$ components.  
In some situation, we set $\mu_{2s} = \Sigma^{1/2}\mu_2 = (\mathbf{\Delta}_l^T, \mathbf{0}_{p-l}^T)$. Throughout the simulation, we set $p=500$ with the number of training data and test data as $n_{1, train} = n_{2, train} = 50$ and $n_{1, test} = n_{2, test} = 250$.

We consider the following 
simulation settings for various configurations of 
$l$ and $\Delta$  as well as different structure of $\Sigma$. 

\begin{enumerate}
    \item Setting 1: AR(1) structured precision matrix. We set $\Sigma^{-1}$ as $(\Sigma^{-1})_{ij} = \rho^{|i-j|}$. 
    
    -simulation 1-1: $\Delta = 3, l=20$ in $\mu_{2s}$ with $\rho = 0.8$.
    
    -simulation 1-2: $\Delta = 0.4, l=400$ in $\mu_{2s}$ with $\rho = 0.8$.
    
    \item Setting 2: Blocked AR(1) structured precision matrix. We set $\Sigma^{-1}$ as $ 
    \begin{bmatrix} 
 \Sigma_q^{-1}&  0\\  
 0&  I_{p-q}
\end{bmatrix}$, where $(\Sigma_q^{-1})_{ij} = \rho^{|i-j|}$ for $1\le i, j \le q$.

    -simulation 2-1: $\Delta = 3, l=20$ in $\mu_{2s}$ with $\rho = 0.9, q=50$.

    -simulation 2-2: $\Delta = 0.15, l=400$ in $\mu_{2s}$ with $\rho = 0.9, q=50$.
    
    \item Setting 3: Exchangeable precision matrix. We set $\Sigma^{-1}$ as $(\Sigma^{-1})_{ij}$ to have 1 as a diagonal components and $\rho$ as all off-diagonal elements.
    
    -simulation 3-1: $\Delta = 0.7, l=20$ in $\mu_{2}$ with $\rho = 0.3$.
    
    -simulation 3-2: $\Delta = 0.07, l=20$ in $\mu_{2s}$ with $\rho = 0.3$.
    
    \item Setting 4: Toeplitz covariance matrix. We set $\Sigma$ as $\Sigma_{ij} = 1\,/ \,(|i-j|+1)$.
    
    -simulation 4-1: $\Delta = 1, l=20$ in $\mu_{2}$.
    
    -simulation 4-2: $\Delta = 0.4, l=20$ in $\mu_{2s}$.
    
    \item Setting 5: Banded covariance matrix. We set $\Sigma$ as $\Sigma_{ij} = \max(1-|i-j| \,/\, 10, 0)$.
    
    -simulation 5-1: $\Delta = 0.7, l=20$ in $\mu_{2}$.
    
    -simulation 5-2: $\Delta = 0.6, l=20$ in $\mu_{2s}$.
\end{enumerate}

%\subsection{Simulation results}
We present misclassification error rates in Table \ref{tab:sim1} to Table \ref{tab:sim5}, 
%\btext
{in which the columns and rows denote the precision matrix and mean vector estimation methods, respectively. We also provide the error rates obtained when the true precision matrix is used instead of estimators, denoted as `Oracle.prec', to verify the effect of the precision matrix estimation. Additionally, we analyze the error rate with a precision matrix estimator which is diagonal and each entry is an inverse of the sample variance of each feature, denoted as 'IR'. We summarize the result as follows : }
%From Table \ref{tab:simulation result}, we observe the  results 
%from the following view points:  
\begin{itemize}
    \item  \textit{Impact of decorrelation of the LAM on the mean vector estimations:}\\ From all the results from the LAM and the IR, we noticed that mean vector estimation methods have a consistent impact on the error rates. Especially the SM method is improved by the NPMLE, and the NPEB  since the error rates  of the NPEB and the NPMLE methods are smaller than those of the SM.  
    \item  \textit{Impact of decorrelation using glasso on the mean vector estimations:}\\
    When the precision matrix is estimated by %\btext
    {glasso}, the mean vector estimation itself didn't show a clear impact on the error rate. %\btext
    {NPEB and NPMLE methods with the decorrelation using glasso method do not significantly improve the error rate of SM $\&$ glasso based discriminant rule}.  
    \item \textit{Impact of glasso and LAM:} \\
 Depending on the structure of the covariance matrix or  precision matrix, 
 the glasso and the LAM are designed for sparse and general structures, respectively. %which precision estimation method performs better varies. 
 For example, in simulation 2-1, the glasso method tends to produce 
 smaller error rates than those from the LAM method. 
 Simulation 2-1 shows the sparse structure of the precision matrix 
 since the AR(1) structure %\btext
 {is 
 %almost similar to %
 the banded matrix.}
 %\kimcomment{Since the precision matrix of AR(1) model is exactly banded (not almost similar to banded), I modified the sentence.}
 On the other hand,  in situation 3-1, using the precision matrix with all the same off-diagonal terms, the LAM method performs better than glasso method since the precision matrix in simulation 3-1 is a dense matrix.  
 These two simulations coincide with 
 the effect of the glasso and the LAM on the different structures of the precision matrix, such as sparsity and denseness.  
% . Since the precision matrix is sparse in 2-1 and dense in 3-1, glasso performed well in 2-1 and not in 3-1, which coincides with the intuition of glasso method in \cite{friedman2008sparse}. 

    \item \textit{No uniform dominance among precision estimation strategies:}\\
    Comparing the results in simulations 2-1 and 2-2, one can notice that even though two situations are generated with the same precision matrix, a tendency in error rates is dissimilar. In situation 2-1, the linear discriminant rule using the graphical lasso method shows the lowest error rate with every mean estimation strategy, which is completely opposite to simulation 2-2. 
    Therefore, the dominant precision estimation strategy for the linear discriminant rule cannot be determined depending on the precision matrix only.
   % \parkcomment{Not clear.  Compared to what, the error rates increase or decrease?}
%    \parkcomment{"Interaction" is not a clear terminology in this context.}
 %   \kimcomment{Resolved: changed the statement}
  %  \item Impact of discriminant rule 
   % \\
%    In most of the cases, the discriminant rules based on (\ref{eqn:discriminant rule - mean each}) and (\ref{eqn:discriminant rule - mean diff}) showed comparable error rates. However, in some cases, there seem to be differences between these two methods.  
%    For example, the results from the LAM in simulation 1-1 show that  the NPEB1 and NPMLE1 methods  outperform the NPEB2 and NPMLE2 
 %   and the opposite results  in simulation 1-2 happen with the LAM method.  
    \item \textit{Comparison with theoretical result:}\\
    The theoretical result can elucidate the significant dominance of the NPMLE and NPEB method over the SM method in simulations 1-2 and 2-2 under the IR and LAM method for the estimation of the precision matrix. According to theorem \ref{thm:NPEB}, when there is a dense signal with low signal intensity ($5/6 < a < 1,\, b<0$), the NPMLE and NPEB methods ensure a wider range of $b$ in which the asymptotic error rate is zero compared to SM method. One can check that simulation 1-2 and 2-2 nearly falls into this area.
 %   \parkcomment{This situation should be true for simulations 2-2, 3-2, 4-2 and 5-2 since $\Delta$ is small, but we do not have the expected result. Only in simulatin 1-2,   the NPMLE with the LAM has advantage.      }
%    \kimcomment{Resolved: dense signal case is only deal tin simulation 1-2 and 2-2. Since NPMLE does not show dominance over NPEB method in situation 2-2, I changed the statement to argue the dominance of `NPEB and NPMLE' over 'SM' method.}
    
    % \item impact of block information in covariance matrix \\
    % When we compare the results from 2-1 and 2-2 according to the block information, we can observe that linear discriminant rules built from LAM method is significantly improved by using block information. Since LAM method estimates all the elements of the precision matrix, when it is block-shaped, using block information can eliminate tentative noises. In contrast, block information does not seem conducive to glasso method. Since glasso method inherently has the feature of shrinking the components which are regarded as noise to be zero, the block information acts as an additional, not necessarily helpful, information. 
\item \textit{Impact of mean vector estimation by shrinkage:} \\
From the results in simulations 3-1 and 3-2, we can observe conspicuous improvements in the SM method in error rates by using the NPMLE and NPEB methods 
for any given estimation of the precision matrix.  
%These results are aligned with the theoretical results in the previous section. 
When the sparse signals have an intensity that is proportional to the positive power of $p$, namely, when $0<a<1/2$, $0<b<1/4$, and $a + 2b < 1/2$ holds, the linear discriminant rules  built from the SM method are known to asymptotically act as random guessing, while the discriminant rules using NPEB and NPMLE have the error rate which asymptotically converges to zero.

\end{itemize}

\begin{table}[p]
\centering
\resizebox{\columnwidth}{!}{%
\begin{tabular}{c|cccc|cccc|}
 & \multicolumn{4}{c|}{Simulation 1-1} & \multicolumn{4}{c|}{Simulation 1-2} \\ \cline{2-9} 
 & \multicolumn{1}{c|}{Oracle.prec} & \multicolumn{1}{c|}{glasso} & \multicolumn{1}{c|}{LAM} & IR & \multicolumn{1}{c|}{Oracle.prec} & \multicolumn{1}{c|}{glasso} & \multicolumn{1}{c|}{LAM} & IR \\ \hline
NPEB1 & \multicolumn{1}{c|}{\begin{tabular}[c]{@{}c@{}}0.0000\\ (0.0000)\end{tabular}} & \multicolumn{1}{c|}{\begin{tabular}[c]{@{}c@{}}0.1374\\ (0.0156)\end{tabular}} & \multicolumn{1}{c|}{\begin{tabular}[c]{@{}c@{}}0.0468\\ (0.0093)\end{tabular}} & \begin{tabular}[c]{@{}c@{}}0.0644\\ (0.0113)\end{tabular} & \multicolumn{1}{c|}{\begin{tabular}[c]{@{}c@{}}0.0001\\ (0.0004)\end{tabular}} & \multicolumn{1}{c|}{\begin{tabular}[c]{@{}c@{}}0.4079\\ (0.0215)\end{tabular}} & \multicolumn{1}{c|}{\begin{tabular}[c]{@{}c@{}}0.3236\\ (0.0200)\end{tabular}} & \begin{tabular}[c]{@{}c@{}}0.3053\\ (0.0201)\end{tabular} \\
NPEB2 & \multicolumn{1}{c|}{\begin{tabular}[c]{@{}c@{}}0.0000\\ (0.0000)\end{tabular}} & \multicolumn{1}{c|}{\begin{tabular}[c]{@{}c@{}}0.1375\\ (0.0158)\end{tabular}} & \multicolumn{1}{c|}{\begin{tabular}[c]{@{}c@{}}0.1139\\ (0.0144)\end{tabular}} & \begin{tabular}[c]{@{}c@{}}0.1146\\ (0.0146)\end{tabular} & \multicolumn{1}{c|}{\begin{tabular}[c]{@{}c@{}}0.0001\\ (0.0004)\end{tabular}} & \multicolumn{1}{c|}{\begin{tabular}[c]{@{}c@{}}0.4084\\ (0.0219)\end{tabular}} & \multicolumn{1}{c|}{\begin{tabular}[c]{@{}c@{}}0.3099\\ (0.0206)\end{tabular}} & \begin{tabular}[c]{@{}c@{}}0.2840\\ (0.0196)\end{tabular} \\
NPMLE1 & \multicolumn{1}{c|}{\begin{tabular}[c]{@{}c@{}}0.0000\\ (0.0000)\end{tabular}} & \multicolumn{1}{c|}{\begin{tabular}[c]{@{}c@{}}0.1370\\ (0.0157)\end{tabular}} & \multicolumn{1}{c|}{\begin{tabular}[c]{@{}c@{}}0.0324\\ (0.0078)\end{tabular}} & \begin{tabular}[c]{@{}c@{}}0.0520\\ (0.0101)\end{tabular} & \multicolumn{1}{c|}{\begin{tabular}[c]{@{}c@{}}0.0000\\ (0.0000)\end{tabular}} & \multicolumn{1}{c|}{\begin{tabular}[c]{@{}c@{}}0.4073\\ (0.0217)\end{tabular}} & \multicolumn{1}{c|}{\begin{tabular}[c]{@{}c@{}}0.2288\\ (0.0183)\end{tabular}} & \begin{tabular}[c]{@{}c@{}}0.1422\\ (0.0162)\end{tabular} \\
NPMLE2 & \multicolumn{1}{c|}{\begin{tabular}[c]{@{}c@{}}0.0000\\ (0.0000)\end{tabular}} & \multicolumn{1}{c|}{\begin{tabular}[c]{@{}c@{}}0.1375\\ (0.0156)\end{tabular}} & \multicolumn{1}{c|}{\begin{tabular}[c]{@{}c@{}}0.0939\\ (0.0130)\end{tabular}} & \begin{tabular}[c]{@{}c@{}}0.0988\\ (0.0138)\end{tabular} & \multicolumn{1}{c|}{\begin{tabular}[c]{@{}c@{}}0.0000\\ (0.0000)\end{tabular}} & \multicolumn{1}{c|}{\begin{tabular}[c]{@{}c@{}}0.4079\\ (0.0218)\end{tabular}} & \multicolumn{1}{c|}{\begin{tabular}[c]{@{}c@{}}0.1379\\ (0.0148)\end{tabular}} & \begin{tabular}[c]{@{}c@{}}0.0489\\ (0.0099)\end{tabular} \\
SM & \multicolumn{1}{c|}{\begin{tabular}[c]{@{}c@{}}0.0000\\ (0.0000)\end{tabular}} & \multicolumn{1}{c|}{\begin{tabular}[c]{@{}c@{}}0.1393\\ (0.0157)\end{tabular}} & \multicolumn{1}{c|}{\begin{tabular}[c]{@{}c@{}}0.1945\\ (0.0178)\end{tabular}} & \begin{tabular}[c]{@{}c@{}}0.1854\\ (0.0177)\end{tabular} & \multicolumn{1}{c|}{\begin{tabular}[c]{@{}c@{}}0.0003\\ (0.0007)\end{tabular}} & \multicolumn{1}{c|}{\begin{tabular}[c]{@{}c@{}}0.4081\\ (0.0220)\end{tabular}} & \multicolumn{1}{c|}{\begin{tabular}[c]{@{}c@{}}0.4352\\ (0.0213)\end{tabular}} & \begin{tabular}[c]{@{}c@{}}0.4344\\ (0.0212)\end{tabular} \\ \hline
\end{tabular}%
}
\caption{Simulation 1 result}
\label{tab:sim1}
\end{table}

\begin{table}[p]
\centering
\resizebox{\columnwidth}{!}{%
\begin{tabular}{c|cccc|cccc|}
 & \multicolumn{4}{c|}{Simulation 2-1} & \multicolumn{4}{c|}{Simulation 2-2} \\ \cline{2-9} 
 & \multicolumn{1}{c|}{Oracle.prec} & \multicolumn{1}{c|}{glasso} & \multicolumn{1}{c|}{LAM} & IR & \multicolumn{1}{c|}{Oracle.prec} & \multicolumn{1}{c|}{glasso} & \multicolumn{1}{c|}{LAM} & IR \\ \hline
NPEB1 & \multicolumn{1}{c|}{\begin{tabular}[c]{@{}c@{}}0.0000\\ (0.0000)\end{tabular}} & \multicolumn{1}{c|}{\begin{tabular}[c]{@{}c@{}}0.0690\\ (0.0115)\end{tabular}} & \multicolumn{1}{c|}{\begin{tabular}[c]{@{}c@{}}0.1235\\ (0.0147)\end{tabular}} & \begin{tabular}[c]{@{}c@{}}0.1692\\ (0.0172)\end{tabular} & \multicolumn{1}{c|}{\begin{tabular}[c]{@{}c@{}}0.0952\\ (0.0132)\end{tabular}} & \multicolumn{1}{c|}{\begin{tabular}[c]{@{}c@{}}0.3184\\ (0.0211)\end{tabular}} & \multicolumn{1}{c|}{\begin{tabular}[c]{@{}c@{}}0.1997\\ (0.0179)\end{tabular}} & \begin{tabular}[c]{@{}c@{}}0.1205\\ (0.0136)\end{tabular} \\
NPEB2 & \multicolumn{1}{c|}{\begin{tabular}[c]{@{}c@{}}0.0000\\ (0.0000)\end{tabular}} & \multicolumn{1}{c|}{\begin{tabular}[c]{@{}c@{}}0.0686\\ (0.0115)\end{tabular}} & \multicolumn{1}{c|}{\begin{tabular}[c]{@{}c@{}}0.1522\\ (0.0156)\end{tabular}} & \begin{tabular}[c]{@{}c@{}}0.1692\\ (0.0169)\end{tabular} & \multicolumn{1}{c|}{\begin{tabular}[c]{@{}c@{}}0.0946\\ (0.0128)\end{tabular}} & \multicolumn{1}{c|}{\begin{tabular}[c]{@{}c@{}}0.3191\\ (0.0206)\end{tabular}} & \multicolumn{1}{c|}{\begin{tabular}[c]{@{}c@{}}0.1708\\ (0.0162)\end{tabular}} & \begin{tabular}[c]{@{}c@{}}0.1168\\ (0.0139)\end{tabular} \\
NPMLE1 & \multicolumn{1}{c|}{\begin{tabular}[c]{@{}c@{}}0.0000\\ (0.0000)\end{tabular}} & \multicolumn{1}{c|}{\begin{tabular}[c]{@{}c@{}}0.0687\\ (0.0116)\end{tabular}} & \multicolumn{1}{c|}{\begin{tabular}[c]{@{}c@{}}0.1199\\ (0.0148)\end{tabular}} & \begin{tabular}[c]{@{}c@{}}0.1609\\ (0.0164)\end{tabular} & \multicolumn{1}{c|}{\begin{tabular}[c]{@{}c@{}}0.0868\\ (0.0126)\end{tabular}} & \multicolumn{1}{c|}{\begin{tabular}[c]{@{}c@{}}0.3196\\ (0.0211)\end{tabular}} & \multicolumn{1}{c|}{\begin{tabular}[c]{@{}c@{}}0.2036\\ (0.0167)\end{tabular}} & \begin{tabular}[c]{@{}c@{}}0.1112\\ (0.0139)\end{tabular} \\
NPMLE2 & \multicolumn{1}{c|}{\begin{tabular}[c]{@{}c@{}}0.0000\\ (0.0000)\end{tabular}} & \multicolumn{1}{c|}{\begin{tabular}[c]{@{}c@{}}0.0688\\ (0.0115)\end{tabular}} & \multicolumn{1}{c|}{\begin{tabular}[c]{@{}c@{}}0.1534\\ (0.0164)\end{tabular}} & \begin{tabular}[c]{@{}c@{}}0.1665\\ (0.0167)\end{tabular} & \multicolumn{1}{c|}{\begin{tabular}[c]{@{}c@{}}0.0883\\ (0.0131)\end{tabular}} & \multicolumn{1}{c|}{\begin{tabular}[c]{@{}c@{}}0.3192\\ (0.0211)\end{tabular}} & \multicolumn{1}{c|}{\begin{tabular}[c]{@{}c@{}}0.2197\\ (0.0189)\end{tabular}} & \begin{tabular}[c]{@{}c@{}}0.1032\\ (0.0135)\end{tabular} \\
SM & \multicolumn{1}{c|}{\begin{tabular}[c]{@{}c@{}}0.0000\\ (0.0000)\end{tabular}} & \multicolumn{1}{c|}{\begin{tabular}[c]{@{}c@{}}0.0708\\ (0.0116)\end{tabular}} & \multicolumn{1}{c|}{\begin{tabular}[c]{@{}c@{}}0.1665\\ (0.0157)\end{tabular}} & \begin{tabular}[c]{@{}c@{}}0.2682\\ (0.0198)\end{tabular} & \multicolumn{1}{c|}{\begin{tabular}[c]{@{}c@{}}0.2002\\ (0.0181)\end{tabular}} & \multicolumn{1}{c|}{\begin{tabular}[c]{@{}c@{}}0.3208\\ (0.0208)\end{tabular}} & \multicolumn{1}{c|}{\begin{tabular}[c]{@{}c@{}}0.2957\\ (0.0203)\end{tabular}} & \begin{tabular}[c]{@{}c@{}}0.2253\\ (0.0190)\end{tabular} \\ \hline
\end{tabular}%
}
\caption{Simulation 2 result}
\label{tab:sim2}
\end{table}

% \usepackage{graphicx}
\begin{table}[p]
\centering
\resizebox{\columnwidth}{!}{%
\begin{tabular}{c|cccc|cccc|}
 & \multicolumn{4}{c|}{Simulation 3-1} & \multicolumn{4}{c|}{Simulation 3-2} \\ \cline{2-9} 
 & \multicolumn{1}{c|}{Oracle.prec} & \multicolumn{1}{c|}{glasso} & \multicolumn{1}{c|}{LAM} & IR & \multicolumn{1}{c|}{Oracle.prec} & \multicolumn{1}{c|}{glasso} & \multicolumn{1}{c|}{LAM} & IR \\ \hline
NPEB1 & \multicolumn{1}{c|}{\begin{tabular}[c]{@{}c@{}}0.0000\\ (0.0000)\end{tabular}} & \multicolumn{1}{c|}{\begin{tabular}[c]{@{}c@{}}0.2609\\ (0.0194)\end{tabular}} & \multicolumn{1}{c|}{\begin{tabular}[c]{@{}c@{}}0.1285\\ (0.0145)\end{tabular}} & \begin{tabular}[c]{@{}c@{}}0.1291\\ (0.0144)\end{tabular} & \multicolumn{1}{c|}{\begin{tabular}[c]{@{}c@{}}0.0720\\ (0.0112)\end{tabular}} & \multicolumn{1}{c|}{\begin{tabular}[c]{@{}c@{}}0.2042\\ (0.0170)\end{tabular}} & \multicolumn{1}{c|}{\begin{tabular}[c]{@{}c@{}}0.0830\\ (0.0120)\end{tabular}} & \begin{tabular}[c]{@{}c@{}}0.0834\\ (0.0122)\end{tabular} \\
NPEB2 & \multicolumn{1}{c|}{\begin{tabular}[c]{@{}c@{}}0.0000\\ (0.0000)\end{tabular}} & \multicolumn{1}{c|}{\begin{tabular}[c]{@{}c@{}}0.2629\\ (0.0198)\end{tabular}} & \multicolumn{1}{c|}{\begin{tabular}[c]{@{}c@{}}0.1874\\ (0.0174)\end{tabular}} & \begin{tabular}[c]{@{}c@{}}0.1887\\ (0.0172)\end{tabular} & \multicolumn{1}{c|}{\begin{tabular}[c]{@{}c@{}}0.0992\\ (0.0130)\end{tabular}} & \multicolumn{1}{c|}{\begin{tabular}[c]{@{}c@{}}0.2048\\ (0.0173)\end{tabular}} & \multicolumn{1}{c|}{\begin{tabular}[c]{@{}c@{}}0.1140\\ (0.0139)\end{tabular}} & \begin{tabular}[c]{@{}c@{}}0.1179\\ (0.0143)\end{tabular} \\
NPMLE1 & \multicolumn{1}{c|}{\begin{tabular}[c]{@{}c@{}}0.0000\\ (0.0000)\end{tabular}} & \multicolumn{1}{c|}{\begin{tabular}[c]{@{}c@{}}0.2611\\ (0.0196)\end{tabular}} & \multicolumn{1}{c|}{\begin{tabular}[c]{@{}c@{}}0.1130\\ (0.0137)\end{tabular}} & \begin{tabular}[c]{@{}c@{}}0.1154\\ (0.0142)\end{tabular} & \multicolumn{1}{c|}{\begin{tabular}[c]{@{}c@{}}0.0648\\ (0.0110)\end{tabular}} & \multicolumn{1}{c|}{\begin{tabular}[c]{@{}c@{}}0.2047\\ (0.0168)\end{tabular}} & \multicolumn{1}{c|}{\begin{tabular}[c]{@{}c@{}}0.0724\\ (0.0113)\end{tabular}} & \begin{tabular}[c]{@{}c@{}}0.0728\\ (0.0114)\end{tabular} \\
NPMLE2 & \multicolumn{1}{c|}{\begin{tabular}[c]{@{}c@{}}0.0000\\ (0.0000)\end{tabular}} & \multicolumn{1}{c|}{\begin{tabular}[c]{@{}c@{}}0.2625\\ (0.0199)\end{tabular}} & \multicolumn{1}{c|}{\begin{tabular}[c]{@{}c@{}}0.1705\\ (0.0169)\end{tabular}} & \begin{tabular}[c]{@{}c@{}}0.1784\\ (0.0167)\end{tabular} & \multicolumn{1}{c|}{\begin{tabular}[c]{@{}c@{}}0.0897\\ (0.0126)\end{tabular}} & \multicolumn{1}{c|}{\begin{tabular}[c]{@{}c@{}}0.2041\\ (0.0173)\end{tabular}} & \multicolumn{1}{c|}{\begin{tabular}[c]{@{}c@{}}0.1019\\ (0.0129)\end{tabular}} & \begin{tabular}[c]{@{}c@{}}0.1094\\ (0.0141)\end{tabular} \\
SM & \multicolumn{1}{c|}{\begin{tabular}[c]{@{}c@{}}0.0000\\ (0.0000)\end{tabular}} & \multicolumn{1}{c|}{\begin{tabular}[c]{@{}c@{}}0.2656\\ (0.0200)\end{tabular}} & \multicolumn{1}{c|}{\begin{tabular}[c]{@{}c@{}}0.2550\\ (0.0194)\end{tabular}} & \begin{tabular}[c]{@{}c@{}}0.2503\\ (0.0190)\end{tabular} & \multicolumn{1}{c|}{\begin{tabular}[c]{@{}c@{}}0.1861\\ (0.0171)\end{tabular}} & \multicolumn{1}{c|}{\begin{tabular}[c]{@{}c@{}}0.2088\\ (0.0172)\end{tabular}} & \multicolumn{1}{c|}{\begin{tabular}[c]{@{}c@{}}0.1977\\ (0.0173)\end{tabular}} & \begin{tabular}[c]{@{}c@{}}0.1936\\ (0.0173)\end{tabular} \\ \hline
\end{tabular}%
}
\caption{Simulation 3 result}
\label{tab:sim3}
\end{table}

\begin{table}[p]
\centering
\resizebox{\columnwidth}{!}{%
\begin{tabular}{c|cccc|cccc|}
 & \multicolumn{4}{c|}{Simulation 4-1} & \multicolumn{4}{c|}{Simulation 4-2} \\ \cline{2-9} 
 & \multicolumn{1}{c|}{Oracle.prec} & \multicolumn{1}{c|}{glasso} & \multicolumn{1}{c|}{LAM} & IR & \multicolumn{1}{c|}{Oracle.prec} & \multicolumn{1}{c|}{glasso} & \multicolumn{1}{c|}{LAM} & IR \\ \hline
NPEB1 & \multicolumn{1}{c|}{\begin{tabular}[c]{@{}c@{}}0.1888\\ (0.0172)\end{tabular}} & \multicolumn{1}{c|}{\begin{tabular}[c]{@{}c@{}}0.1990\\ (0.0178)\end{tabular}} & \multicolumn{1}{c|}{\begin{tabular}[c]{@{}c@{}}0.1554\\ (0.0155)\end{tabular}} & \begin{tabular}[c]{@{}c@{}}0.1625\\ (0.0159)\end{tabular} & \multicolumn{1}{c|}{\begin{tabular}[c]{@{}c@{}}0.2601\\ (0.0198)\end{tabular}} & \multicolumn{1}{c|}{\begin{tabular}[c]{@{}c@{}}0.2523\\ (0.0190)\end{tabular}} & \multicolumn{1}{c|}{\begin{tabular}[c]{@{}c@{}}0.2024\\ (0.0171)\end{tabular}} & \begin{tabular}[c]{@{}c@{}}0.2107\\ (0.0176)\end{tabular} \\
NPEB2 & \multicolumn{1}{c|}{\begin{tabular}[c]{@{}c@{}}0.2424\\ (0.0185)\end{tabular}} & \multicolumn{1}{c|}{\begin{tabular}[c]{@{}c@{}}0.1985\\ (0.0177)\end{tabular}} & \multicolumn{1}{c|}{\begin{tabular}[c]{@{}c@{}}0.1612\\ (0.0164)\end{tabular}} & \begin{tabular}[c]{@{}c@{}}0.1684\\ (0.0168)\end{tabular} & \multicolumn{1}{c|}{\begin{tabular}[c]{@{}c@{}}0.3314\\ (0.0212)\end{tabular}} & \multicolumn{1}{c|}{\begin{tabular}[c]{@{}c@{}}0.2527\\ (0.0192)\end{tabular}} & \multicolumn{1}{c|}{\begin{tabular}[c]{@{}c@{}}0.2150\\ (0.0180)\end{tabular}} & \begin{tabular}[c]{@{}c@{}}0.2191\\ (0.0184)\end{tabular} \\
NPMLE1 & \multicolumn{1}{c|}{\begin{tabular}[c]{@{}c@{}}0.1691\\ (0.0159)\end{tabular}} & \multicolumn{1}{c|}{\begin{tabular}[c]{@{}c@{}}0.1983\\ (0.0178)\end{tabular}} & \multicolumn{1}{c|}{\begin{tabular}[c]{@{}c@{}}0.1536\\ (0.0155)\end{tabular}} & \begin{tabular}[c]{@{}c@{}}0.1610\\ (0.0162)\end{tabular} & \multicolumn{1}{c|}{\begin{tabular}[c]{@{}c@{}}0.2385\\ (0.0187)\end{tabular}} & \multicolumn{1}{c|}{\begin{tabular}[c]{@{}c@{}}0.2524\\ (0.0191)\end{tabular}} & \multicolumn{1}{c|}{\begin{tabular}[c]{@{}c@{}}0.2020\\ (0.0175)\end{tabular}} & \begin{tabular}[c]{@{}c@{}}0.1610\\ (0.0162)\end{tabular} \\
NPMLE2 & \multicolumn{1}{c|}{\begin{tabular}[c]{@{}c@{}}0.2241\\ (0.0182)\end{tabular}} & \multicolumn{1}{c|}{\begin{tabular}[c]{@{}c@{}}0.1985\\ (0.0176)\end{tabular}} & \multicolumn{1}{c|}{\begin{tabular}[c]{@{}c@{}}0.1579\\ (0.0161)\end{tabular}} & \begin{tabular}[c]{@{}c@{}}0.1666\\ (0.0167)\end{tabular} & \multicolumn{1}{c|}{\begin{tabular}[c]{@{}c@{}}0.3215\\ (0.0217)\end{tabular}} & \multicolumn{1}{c|}{\begin{tabular}[c]{@{}c@{}}0.2530\\ (0.0192)\end{tabular}} & \multicolumn{1}{c|}{\begin{tabular}[c]{@{}c@{}}0.2127\\ (0.0185)\end{tabular}} & \begin{tabular}[c]{@{}c@{}}0.1666\\ (0.0167)\end{tabular} \\
SM & \multicolumn{1}{c|}{\begin{tabular}[c]{@{}c@{}}0.2998\\ (0.0208)\end{tabular}} & \multicolumn{1}{c|}{\begin{tabular}[c]{@{}c@{}}0.2015\\ (0.0180)\end{tabular}} & \multicolumn{1}{c|}{\begin{tabular}[c]{@{}c@{}}0.1944\\ (0.0172)\end{tabular}} & \begin{tabular}[c]{@{}c@{}}0.1919\\ (0.0169)\end{tabular} & \multicolumn{1}{c|}{\begin{tabular}[c]{@{}c@{}}0.3660\\ (0.0217)\end{tabular}} & \multicolumn{1}{c|}{\begin{tabular}[c]{@{}c@{}}0.2558\\ (0.0192)\end{tabular}} & \multicolumn{1}{c|}{\begin{tabular}[c]{@{}c@{}}0.2448\\ (0.0187)\end{tabular}} & \begin{tabular}[c]{@{}c@{}}0.2422\\ (0.0190)\end{tabular} \\ \hline
\end{tabular}%
}
\caption{Simulation 4 result}
\label{tab:sim4}
\end{table}

\begin{table}[p]
\centering
\resizebox{\columnwidth}{!}{%
\begin{tabular}{c|cccc|cccc|}
 & \multicolumn{4}{c|}{Simulation 5-1} & \multicolumn{4}{c|}{Simulation 5-2} \\ \cline{2-9} 
 & \multicolumn{1}{c|}{Oracle.prec} & \multicolumn{1}{c|}{glasso} & \multicolumn{1}{c|}{LAM} & IR & \multicolumn{1}{c|}{Oracle.prec} & \multicolumn{1}{c|}{glasso} & \multicolumn{1}{c|}{LAM} & IR \\ \hline
NPEB1 & \multicolumn{1}{c|}{\begin{tabular}[c]{@{}c@{}}0.0960\\ (0.0134)\end{tabular}} & \multicolumn{1}{c|}{\begin{tabular}[c]{@{}c@{}}0.3562\\ (0.0209)\end{tabular}} & \multicolumn{1}{c|}{\begin{tabular}[c]{@{}c@{}}0.3077\\ (0.0203)\end{tabular}} & \begin{tabular}[c]{@{}c@{}}0.3055\\ (0.0215)\end{tabular} & \multicolumn{1}{c|}{\begin{tabular}[c]{@{}c@{}}0.1164\\ (0.0144)\end{tabular}} & \multicolumn{1}{c|}{\begin{tabular}[c]{@{}c@{}}0.1346\\ (0.0147)\end{tabular}} & \multicolumn{1}{c|}{\begin{tabular}[c]{@{}c@{}}0.1060\\ (0.0133)\end{tabular}} & \begin{tabular}[c]{@{}c@{}}0.1004\\ (0.0131)\end{tabular} \\
NPEB2 & \multicolumn{1}{c|}{\begin{tabular}[c]{@{}c@{}}0.1250\\ (0.0146)\end{tabular}} & \multicolumn{1}{c|}{\begin{tabular}[c]{@{}c@{}}0.3592\\ (0.0207)\end{tabular}} & \multicolumn{1}{c|}{\begin{tabular}[c]{@{}c@{}}0.3224\\ (0.0211)\end{tabular}} & \begin{tabular}[c]{@{}c@{}}0.3190\\ (0.0212)\end{tabular} & \multicolumn{1}{c|}{\begin{tabular}[c]{@{}c@{}}0.1732\\ (0.0165)\end{tabular}} & \multicolumn{1}{c|}{\begin{tabular}[c]{@{}c@{}}0.1363\\ (0.0149)\end{tabular}} & \multicolumn{1}{c|}{\begin{tabular}[c]{@{}c@{}}0.1124\\ (0.0134)\end{tabular}} & \begin{tabular}[c]{@{}c@{}}0.1026\\ (0.0135)\end{tabular} \\
NPMLE1 & \multicolumn{1}{c|}{\begin{tabular}[c]{@{}c@{}}0.0926\\ (0.0130)\end{tabular}} & \multicolumn{1}{c|}{\begin{tabular}[c]{@{}c@{}}0.3559\\ (0.0212)\end{tabular}} & \multicolumn{1}{c|}{\begin{tabular}[c]{@{}c@{}}0.3059\\ (0.0203)\end{tabular}} & \begin{tabular}[c]{@{}c@{}}0.3005\\ (0.0210)\end{tabular} & \multicolumn{1}{c|}{\begin{tabular}[c]{@{}c@{}}0.1040\\ (0.0131)\end{tabular}} & \multicolumn{1}{c|}{\begin{tabular}[c]{@{}c@{}}0.1341\\ (0.0147)\end{tabular}} & \multicolumn{1}{c|}{\begin{tabular}[c]{@{}c@{}}0.1045\\ (0.0133)\end{tabular}} & \begin{tabular}[c]{@{}c@{}}0.0992\\ (0.0130)\end{tabular} \\
NPMLE2 & \multicolumn{1}{c|}{\begin{tabular}[c]{@{}c@{}}0.1131\\ (0.0143)\end{tabular}} & \multicolumn{1}{c|}{\begin{tabular}[c]{@{}c@{}}0.3579\\ (0.0207)\end{tabular}} & \multicolumn{1}{c|}{\begin{tabular}[c]{@{}c@{}}0.3206\\ (0.0209)\end{tabular}} & \begin{tabular}[c]{@{}c@{}}0.3135\\ (0.0207)\end{tabular} & \multicolumn{1}{c|}{\begin{tabular}[c]{@{}c@{}}0.1611\\ (0.0159)\end{tabular}} & \multicolumn{1}{c|}{\begin{tabular}[c]{@{}c@{}}0.1358\\ (0.0150)\end{tabular}} & \multicolumn{1}{c|}{\begin{tabular}[c]{@{}c@{}}0.1112\\ (0.0134)\end{tabular}} & \begin{tabular}[c]{@{}c@{}}0.1019\\ (0.0133)\end{tabular} \\
SM & \multicolumn{1}{c|}{\begin{tabular}[c]{@{}c@{}}0.2020\\ (0.0174)\end{tabular}} & \multicolumn{1}{c|}{\begin{tabular}[c]{@{}c@{}}0.3633\\ (0.0212)\end{tabular}} & \multicolumn{1}{c|}{\begin{tabular}[c]{@{}c@{}}0.3667\\ (0.0207)\end{tabular}} & \begin{tabular}[c]{@{}c@{}}0.3464\\ (0.0216)\end{tabular} & \multicolumn{1}{c|}{\begin{tabular}[c]{@{}c@{}}0.2451\\ (0.0194)\end{tabular}} & \multicolumn{1}{c|}{\begin{tabular}[c]{@{}c@{}}0.1397\\ (0.0152)\end{tabular}} & \multicolumn{1}{c|}{\begin{tabular}[c]{@{}c@{}}0.1377\\ (0.0151)\end{tabular}} & \begin{tabular}[c]{@{}c@{}}0.1166\\ (0.0139)\end{tabular} \\ \hline
\end{tabular}%
}
\caption{Simulation 5 result}
\label{tab:sim5}
\end{table}

\section{Real Data Analysis}
\label{sec:Realdata}
In this section, we consider 
five real data sets to   
compare the performance of various binary classification methods discussed in this paper.  Through these real data examples, 
we investigate the role of mean vector and precision matrix estimation in high-dimensional classification.
%Among the examples of such real data, 
%we include examples in the case of block information discussed in section 3. 
%Effects of mean estimation method and precision estimation methods are analyzed according to corresponding misclassification rate.

Four real datasets are used for comparing the performance of discriminant rules. 
\begin{itemize}
    \item EEG data: 122 observations (77: group 1, 45: group 2) with 512 features
    
    We used EEG (Electroencephalography) data to validate linear discriminant rules under multiple estimation methods (which is available at \url{https://archive.ics.uci.edu/ml/datasets/eeg+database)}. The data was initially generated for the large study about the genetic predisposition of alcoholism. Measurements from 64 electrodes placed on the subject's scalps were sampled at 256 Hz for 1 second. For data pre-processing, we parsed this time series data into 8 intervals and extracted each interval's median value. Therefore, we finally made the data of $64 \times 8 = 512$ dimensional vectors from each subject. With 122 observations in sum ($n_1 = 77$ subjects from the alcoholic group and $n_2=45$ from the control group), we examined the performance of each discriminant rule using the LOOCV method.

    \item Gravier data: 168 observations (111: group 1, 57: group 2) with 1000 features
    
    Gene array data obtained for the study pertaining to the prediction of metastasis of small node-negative breast carcinoma is also utilized, which we named Gravier data. Small invasive ductal carcinomas without axillary lymph node involvement (T1T2N0) were analyzed from 168 patients. Among these, 111 subjects who had no events for 5 years afterward were categorized as group 1, and the other 57 subjects with early metastasis were assigned as group 2. We processed the original data by extracting 1000 features that have the most significant t values among the 2905 features beforehand. 

    \item Ham data: 214 observations (111: group 1, 103: group 2) with 431 features.

    Ham data was obtained by food spectrograph about 19 Spanish and 18 French dry-cured hams. A food spectrograph is utilized in chemometrics to classify food types, which can be directly used to assure food safety and quality. This data is publicly available at \url{http://www.timeseriesclassification.com/description.php?Dataset=Ham}. From 37 hams in total, up to 6 observations were obtained. Therefore, 214 observations were categorized into two groups (111 observations for group 1 and 103 for group 2), with 431 feature vectors for each observation. The data preprocessing procedure is described in `Sodium dodecyl sulphate-polyacrylamide gel electrophoresis of proteins in dry-cured hams: Data registration and multivariate analysis across multiple gels' (to be in reference).

    \item IMVigor data: 298 observations (230: group 1, 68: group 2) with 4792 features (divided into 7 groups)
    
    Clinical outcomes of metastatic urothelial 
   % \kimcomment{resolved: fixed the typo `unorthelial' to `urothelial' and added reference for the IMVigor data}
    cancer patients were collected and are available at \url{http://research-pub.gene.com/IMvigor210CoreBiologies/} (see \cite{mariathasan2018tgfbeta}).
    Observations about 230 non-responders and 68 responders are composed of 4792 feature vectors. These features are known to be correlated within seven groups, which contain 1583, 975, 569, 548, 546, 341, and 230 features each. Therefore, we investigated the error rate of linear discriminant rules with and without this group information.

\end{itemize}

Within these data, classification performance has been measured through Leave-One-Out Cross Validation (LOOCV) method. 
As introduced in \ref{sec: construction of discriminant rules}, linear discriminant rules  can be constructed by estimating $\mu_1^*$ and $\mu_2^*$ respectively, or $\mu_2^* - \mu_1^*$ at once. Therefore, we compared the error rate of each discriminant rule according to the mean estimation method, precision estimation method, and discriminant rule construction method.

As one can see from Table \ref{tab:ungrouped real data}, classification errors largely depend on precision matrix estimation methods. A more suitable precision estimation method relies on the data. For example, in HAM data, the LAM method shows a lower error rate, while glasso method performs well with Gravier data. Among mean vector estimation methods, error rates do not vary significantly.  However, when classifying Gravier data with glasso precision estimation method, the error rate varied significantly according to the mean estimation method. Therefore, we compared each component of $\widehat{\mu_2^* - \mu_1^*}$ among mean estimation methods. 

Figure \ref{fig: decor mean comparison result} shows 
the plots of NPEB vs. SM and NPMLE vs. SM after two methods of decorrelation. 
The NPEB and NPMLE tend to have the shrinkage effect of the SM, especially for large SM values. 
It is seen that the   NPMLE in \eqref{eqn:GMLE} with the estimated $\hat G$ has the monotonicity property of SM. 
On the other hand, the NPEB with estimated marginal density $\hat f$ and $\hat f'$ is not guaranteed to have monotonicity, 
so there exist some wiggly patterns in local regions, however, overall patterns of NPEB are similar to those of NPMLE. 
In general, the posterior mean in \eqref{eqn:GMLE} based on $f$-modeling is not guaranteed to have the monotonicity in $z$ while 
the $g$-modeling such as plugging in estimated $\hat G$  into \eqref{eqn:GMLE}   guarantees  the monotonicity in $z$.  
%NPEB method is known for its property that might not preserve the monotonicity, which is counterintuitive. For example, 
See the discussion on this issue in \cite{ignatiadis2022confidence} for the Poisson model, which is also true for the case of normal mean estimation.

%\parkcomment{We may add some reference on the monotonicity of NPMLE and NPEB.}
%\kimcomment{Resolved: added one reference for non-monotonicity of NPEB in Poisson model.}

%, this case shows distinguished discrepancy between sample mean method and Empirical Bayes method. This explains the gap among mean vector estimation methods in this case. 

%\textcolor{red}{In figure (\ref{fig: real data result}), `IR' indicates precision estimation methods using independence rule, and `glasso' denotes glasso method. - duplicated explanation}
%with the parameter $\rho = 0.01$.% 

To summarize,  
the correlation, such as glasso and LAM methods, is effective in most of the cases in that the IR without decorrelation is improved although 
those two methods make the IR worse in some data sets.  
In particular, the glasso is sensitive to different cases since 
the glasso is designed for a sparse structured precision matrix. 
On the other hand, the LAM method is robust to all data sets except Gravier data in which the performance of all classification methods is  worse than the glasso and the IR.  
For the glasso method,  
different classification methods 
have large variations in error  rate while 
the LAM method produces quite similar error rates for different rules.

% \begin{figure}[H]\label{fig:decorrelated mean comparison plot}
% \begin{center}
% \subfloat[Gravier data - glasso  ]{\includegraphics[width = 2.6in]{Gravier_glasso_NPEB.png}} 
% \subfloat[Gravier data - glasso ]{\includegraphics[width = 2.6in]{Gravier_glasso_NPMLE.png}}\\ 
% \subfloat[Gravier data - LAM]{\includegraphics[width = 2.6in]{Gravier_LAM_NPEB.png}}
% \subfloat[Gravier data - LAM]{\includegraphics[width = 2.6in]{Gravier_LAM_NPMLE.png}}\\
% \subfloat[Ham data - glasso ]{\includegraphics[width = 2.6in]{Ham_glasso_NPEB.png}}
% \subfloat[Ham data - glasso ]{\includegraphics[width = 2.6in]{Ham_glasso_NPMLE.png}}\\
% \subfloat[Ham data - LAM]{\includegraphics[width = 2.6in]{Ham_LAM_NPEB.png}}
% \subfloat[Ham data - LAM]{\includegraphics[width = 2.6in]{Ham_LAM_NPMLE.png}}\\
% \caption{$(\widehat{\mu_2^*-\mu_1^*)}$ plot - NPEB, NPMLE method to SM method}
% \label{some example}
% \end{center}
% \end{figure}

\begin{figure}%[H]
    \centering
    \includegraphics[scale = 1]{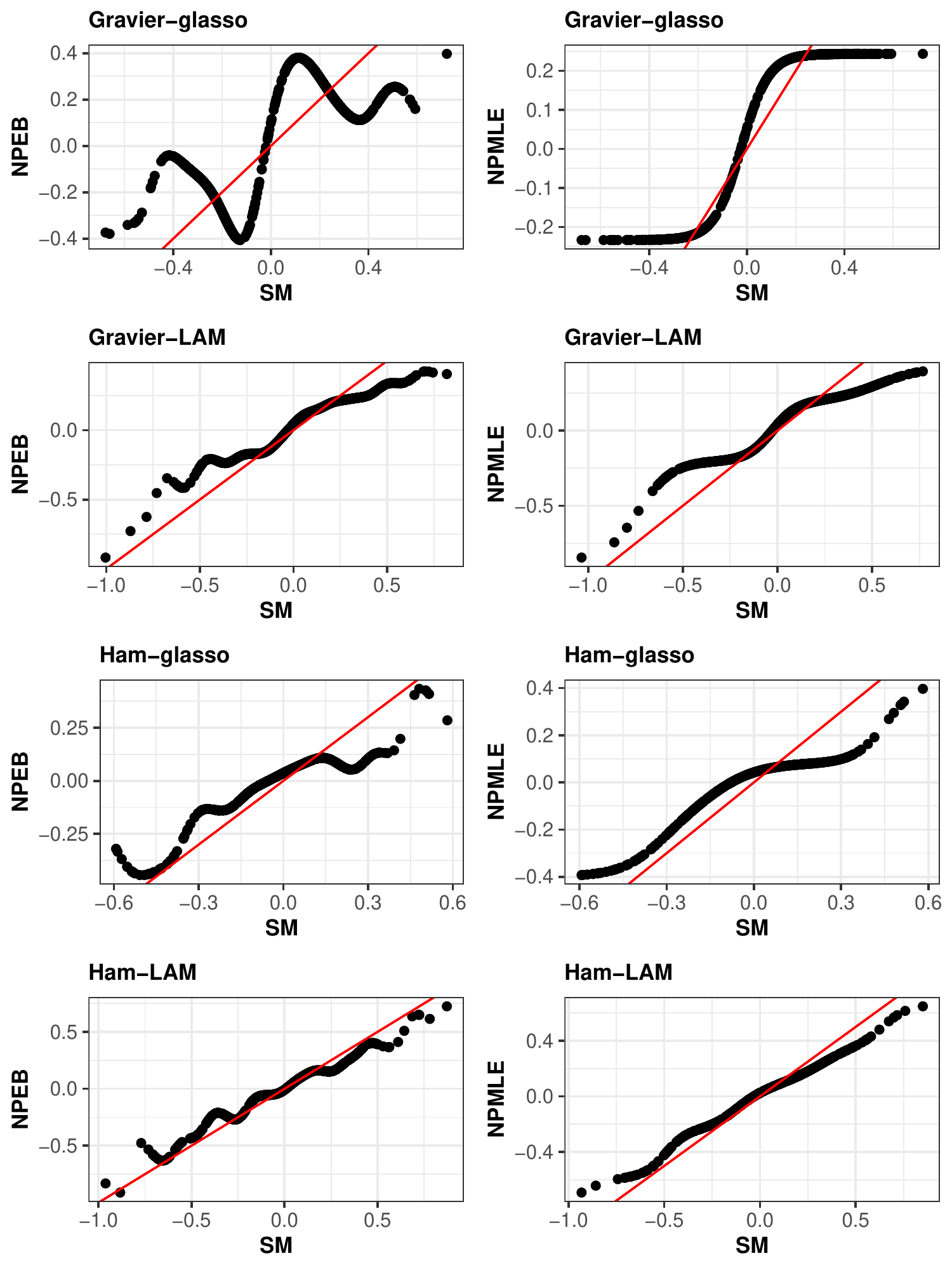}
    \caption{$(\widehat{\mu_2^*-\mu_1^*)}$ plot - NPEB, NPMLE method to SM method}
    \label{fig: decor mean comparison result}
\end{figure}

%\begin{figure}[H]
%    \centering
%    \includegraphics[scale = 1]{real_data_graph.pdf}
%    \caption{Error rate comparison according to mean and precision estimation methods}
%    \label{fig: real data result}
%\end{figure}

%\parkcomment{Figure 3 is redundant due to Table 6, so it will be better to remove Figure 3.}
%\kimcomment{resolved: removed the figure}
%%%%% Table 삽입 - unified

\begin{table}%[H]
\centering
\resizebox{\columnwidth}{!}{%
\begin{tabular}{c|ccc|ccc|ccc|ccc|}
 & \multicolumn{3}{c|}{EEG} & \multicolumn{3}{c|}{Gravier} & \multicolumn{3}{c|}{Ham} & \multicolumn{3}{c|}{IMVigor} \\ \hline
 & glasso & LAM & IR & glasso & LAM & IR & glasso & LAM & IR & glasso & LAM & IR \\ \hline
NPEB1 & *22/122 & 24/122 & 27/122 & *19/168 & 32/168 & 24/168 & 39/214 & 31/214 & 49/214 & 169/298 & 71/298 & 111/298 \\
NPEB2 & 23/122 & 28/122 & 28/122 & 31/168 & 33/168 & 31/168 & 33/214 & *30/214 & 53/214 & 104/298 & *70/298 & 111/298 \\
NPMLE1 & 23/122 & 26/122 & 27/122 & 26/168 & 31/168 & 29/168 & 40/214 & 32/214 & 49/214 & 128/298 & *70/298 & 112/298 \\
NPMLE2 & 24/122 & 26/122 & 27/122 & 30/168 & 33/168 & 31/168 & 33/214 & *30/214 & 52/214 & 106/298 & *70/298 & 110/298 \\
SM & 24/122 & 26/122 & 28/122 & 31/168 & 33/168 & 31/168 & 34/214 & *30/214 & 50/214 & 106/298 & *70/298 & 111/298 \\ \hline
\end{tabular}%
}
\caption{LOOCV classification result about EEG, Gravier, Ham, and IMVigor datasets of corresponding discriminant rules. 
 $'*'$ indicates the lowest error rate in each data set.}
\label{tab:ungrouped real data}
\end{table}
%%%%%%%%%%%%%%%%%%%%%%%%%%%5

\section{Concluding Remarks}
\label{sec:Concluding}
Combining different types of precision matrix and mean vector  estimation strategies, we investigate the performances of 
the various linear discriminant rules. We evaluate these rules
under various contrived settings. With dense and sparse structures of precision matrices and the difference of mean vectors, the performances of the discriminant rules  are investigated. We, therefore, observe that linear discriminant rules perform well when the simulation situation is aligned with the assumptions that each estimation strategy is based on.
Our results including numerical studies and real data examples 
show that 
none of the discriminant rules 
tend to dominate the others.
In particular, we emphasize that 
the theoretical result is presented on the NPEB. 
We believe that this is an interesting result   in the sense that 
we can compare $f$-modeling and $g$-modeling theoretically, which are corresponding to the NPEB and NPMLE where the performance of the NPMLE is studied in \cite{park2022high}.    
%\parkcomment{How do we know this in real data? }
%\kimcomment{Resolved: Moved the last sentence right after the simulation results, to only mention about such result for the simulation settings.}
%\textcolor{red}{To summarize, we }
%In addition, with the classification result for the real data sets, we suggest the availability of multiple mean and precision estimation methods. 
%Especially, we suggest the potential advantage of the linear discriminant rules based on NPEB method theoretically, which were known to be inferior to NPMLE method previously. 
%Though it is limited to the situation with the unified degree of signals, it suggests the accessibility of the NPEB method both in dense and sparse mean vector difference situations. We leave further theoretical analysis of nonparametric mean estimation methods as future work.

One interesting phenomenon is that the structure of the mean vector may be changed after decorrelation. Since our results in this paper show that the different methods of estimation of mean vectors have different performances depending on the structure of mean vectors, so it will be interesting to investigate the changes in mean vector structures after  decorrelation and their effect on the choice of estimation methods. We leave this as future work.  

\section*{Acknowledgement}
Research of J. Park was supported by the National Research Foundation of Korea (NRF) grant funded
by the Korea government (MSIT) (No. 2020R1A2C1A01100526).
\bibliographystyle{acm}
\bibliography{bib_HDM}

\appendix

\section{Lemmas for Theorem 1}
We present the following four lemmas which are used in the proof of Theorem 1. 
\begin{lemma}\label{lemma: bigger than 1-1 over p}
For $p>1$,
\begin{equation}
    \left(1 - \frac{1}{p^2} \right)^p \ge 1 - \cfrac{1}{p}.
\end{equation}
\end{lemma}

\begin{proof}
One can easily show that $f(x) = x\log \left(\cfrac {x+1}{x}\right)$ is an increasing function, by checking
\begin{equation*}
    f'(x) = \log \left(\cfrac {x+1}{x}\right) - \cfrac{1}{x+1},\; f''(x) = - \cfrac{1}{x(x+1)^2} < 0
\end{equation*}
and $\lim_{x \to \infty} f'(x) = 0$. From $f(p) \ge f(p-1)$, one can obtain 
\begin{equation*}
    p\log (1 - \cfrac{1}{p^2}) \ge \log (1 - \frac{1}{p}),
\end{equation*}
which concludes the proof.
\end{proof}

\begin{lemma}\label{lemma: max S_j prob}
Let $Z_1, \cdots, Z_p \overset{\mathrm{iid}}{\sim} N(0, 1)$ and 
$\Delta = p^b$. For $b<0$ and 
the fixed constant $a_n =  \left(1/n_1 + 1/n_2 \right)^{-1/2}$,  
\begin{equation}\label{eqn: max S_j prob}
    P \left(\max_{1\le j \le p} |Z_j| + a_n\Delta < \log p \right) > 1-\frac{1}{p}
\end{equation}
holds for large enough $p$.
\end{lemma}

\begin{proof}
Noting that for all $t>0$, we have  
\begin{equation*}
    \frac{t}{t^2+1}\phi(t) < 1-\Phi(t) \le \frac{1}{t}\phi(t)
\end{equation*}
where $\phi, \Phi$ denote the probability distribution function and the cumulative distribution function of standard normal distribution. Then we obtain

\begin{eqnarray}
    P \left(\max_{1\le j \le p} |Z_j| + a_n\Delta < \log p \right) = \left\{1 - 2\Phi(a_n \Delta - \log p)\right\}^p 
    \ge \left\{1 - \cfrac{2\phi (\log p - a_n \Delta)}{\log p - a_n \Delta}\right\}^p.
    \label{eqn:maxS}
\end{eqnarray}
Since  
$a_n \Delta \to 0$ and  $\log p \to \infty$ as $p \to \infty$, 
 we have $2a_n\Delta < \log p$ for large $p$, 
 therefore  
\eqref{eqn:maxS} satisfies the followings: 
\begin{align*}
\left \{1 - \cfrac{2\phi (\log p - a_n \Delta)}{\log p - a_n \Delta}\right \}^p
 &\ge \left [1 - \cfrac{4\phi \left \{ (\log p) / 2 \right \} }{\log p}\right ]^p \\
 &= 
 \left[ 1 - \cfrac{4}{\sqrt{2\pi} \log p} \exp \left\{-(\log p)^2 / 8 \right\}\right ]^p.
\end{align*}

From lemma \ref{lemma: bigger than 1-1 over p}, 
for large $p$  satisfying 
\begin{equation*}
    \cfrac{4}{\sqrt{2\pi} \log p} \exp \left \{-\cfrac{(\log p)^2}{8} \right \} \le \cfrac{1}{p^2},
\end{equation*}

we obtain
$$P \left(\max_{1\le j \le p} |Z_j| + a_n\Delta < \log p \right) 
 \geq \left(1-\frac{1}{p^2}\right)^p \geq 1-\frac{1}{p}.$$
\end{proof}

\begin{lemma}\label{lemma: maximum prob bound}
Let $Z_i \sim N(0, 1)$ and $\Delta = p^b$. Then, for $b<0$ 
and $a_n =  \left(1/n_1 + 1/n_2 \right)^{-1/2}$, we have  
\begin{equation}
    P\left(|Z_i| + a_n\Delta <\sqrt{2\alpha(\log p + 1)} \right) > 1-\frac{8}{\sqrt{2\alpha\log p}}\frac{e}{\sqrt{2\pi}p^\alpha}
\end{equation}
for large enough $p$.
\end{lemma}

\begin{proof}
From the fact $$\cfrac{t}{t^2+1}\phi(t) < 1-\Phi(t) \le \cfrac{1}{t}\phi(t),$$  for large enough $p$ which satisfy $a_n \sqrt{2\alpha(\log p + 1)} < p^{-b} = 1 / \Delta$ and $\sqrt{2\alpha(\log p + 1)} - 1> \sqrt{2\alpha\log p} \,/\, 2$, we obtain 
\begin{align*}
    &P\left(|Z_i| + a_n\Delta <\sqrt{2\alpha(\log p + 1)} \right) \\ &= 1 - 2\Phi \left(\sqrt{2\alpha(\log p + 1)}  - a_n \Delta\right) \\
    & > 1-\cfrac{2}{\sqrt{2\pi}}\cfrac{2}{\sqrt{2\alpha(\log p + 1)} - a_n\Delta} 
    \exp \left\{-\cfrac{1}{2}(\sqrt{2\alpha(\log p + 1)} - a_n\Delta)^2 \right\} \\
    & > 1 - \cfrac{4}{\sqrt{2\alpha (\log p + 1)} - 1}\cfrac{1}{\sqrt{2\pi}p^\alpha}\exp \left\{ a_n\Delta \sqrt{2\alpha(\log p + 1)} \right\}
    \\& >1-\cfrac{8}{\sqrt{2\alpha\log p}}\cfrac{e}{\sqrt{2\pi}p^\alpha}.
\end{align*}
\end{proof}

%\parkcomment{  $a_n\Delta$ in the proof? \\::resolved (Since $a_n\Delta$ would decay to 0, I ignored this term while making an argument. After including this, the constant is slightly changed. This result would be used in page 36 for $\gamma$, and this does not affect the further content.)}

\begin{lemma} 
For $ \widehat{\mu_{D,i}}^{EB}$ defined in \eqref{eqn: mu hat}, we have the lower bound $W_i$ such that 
$$\sum_{i=1}^{p_1}  \widehat{\mu_{D,i}}^{EB} \geq \sum_{i=1}^{p_1} W_i \geq  \frac{h^2}{1+h^2} a_n p^{a+b}\{1+o(1)\}+o_p\left\{(\log p)^2\right\}$$
where

\begin{eqnarray}
    W_i =
    \begin{cases}
    W_{1i}, & \mbox{if } |Z_i| + a_n\Delta \le \sqrt{2\alpha(\log p + 1)} \\ 
    W_{2i}, & \mbox{o.w}
    \end{cases},
\end{eqnarray} 
and 
\begin{align*}
W_{1i} &= \cfrac{h^2 Z_i^*}{1+h^2} - \cfrac{h^2}{1+h^2}\cfrac{(1 + 2\sqrt{2\pi}p^{1/2+\epsilon})\sqrt{2\alpha(\log p + 1)}  +   {(1 + h^2)/h^4}}
    {p^{1-\alpha}(1-p^{a-1}) - 2p^{1/2+\epsilon}}, \\ 
W_{2i} &=  \Delta + \cfrac{1}{a_n}(1-\log p)Z_i - \cfrac{1}{a_n h^2} \log p.
\end{align*}

Here, $Z_i^* = Z_i + a_n\Delta$.
\label{lemma:lowbound}
\end{lemma}

\begin{proof}

%\btext
{As this proof contains several technical details, we first summarize the flow of the proof. We obtain the first inequality by introducing the auxiliary random variable $V_i$ which satisfies $\widehat{\mu_{D,i}}^{EB} \ge V_i \ge W_i$. Note that the random variable $W_i$ is independent. The second inequality is shown by using the weak law of large numbers (WLLN).}

\begin{itemize}
\item 
First, %a lower bound %
%of $K_i$ in \eqref{eqn:V}%
we consider an auxiliary random variable $V_i$ as the following one :  
\begin{equation}
    V_i = 
    \begin{cases}
    V_{1i}, & \mbox{if } |Z_i| + a_n\Delta \le \sqrt{2\alpha(\log p + 1)} \\ 
    V_{2i}, & \mbox{o.w}
    \end{cases}
\end{equation}

where 
\begin{align*}
    V_{1i} &= \Delta + \cfrac{Z_i}{a_n} - \cfrac{1}{h^2}\cfrac{\cfrac{1}{\sqrt{2\pi}}\cfrac{h^3}{(1+h^2)^{3/2}} \left\{ 
    (p_1-1)Z_i T_i+(p-p_1)Z_i^* T_i^* \right\}
    +2p^{1/2+\epsilon}}
    {\cfrac{1}{\sqrt{2\pi}}+\cfrac{1}{\sqrt{2\pi}}\cfrac{h}{(1+h^2)^{1/2}}
    \left \{(p_1-1)T_i+(p-p_1)T_i^* \right \}
    -2p^{1/2+\epsilon}},\\
    V_{2i} &= \Delta + \cfrac{Z_i}{a_n} + \cfrac{1}{a_n h^2}
    \min{\{(Z_{ji}^-)_{1 \le j \le p_1}, (Z_{ji}^- + a_n \Delta)_{p_1+1 \le j \le p}\}}.
\end{align*}

Here, $Z_{ij}^-, Z_i^*, T_i, T_i^*$ are defined as follows. 
\begin{align*}
   Z_{ij}^- = Z_i - Z_j, \quad  \quad &Z_i^*= Z_i + a_n \Delta, \\
    T_i = \exp \left\{-\frac{Z_i^2}{2(1+h^2)} \right\}, \quad 
    &T_i^* =  \exp \left \{-\frac{Z_i^{*2}}{2(1+h^2)} \right \}.
\end{align*}

We first notice that $V_{i}$ is constructed to always satisfy $ \widehat{\mu_{D,i}}^{EB} \ge  V_{2i}$ (see the derivation of (\ref{eqn:Ki min bound})), and $ \widehat{\mu_{D,i}}^{EB} \ge V_{1i}$ on $E = \cap_{i=1}^{p_1}E_i$.  %\btext
{$E$ is the set of events in which $Z_i$s satisfy a specific relationship, with $P(E)\to 1$ as $ p \to \infty$. Here, we defer the detailed definition of $E$ to Appendix B. This implies that $\widehat{\mu_{D,i}}^{EB}$ can be bounded by $V_{1i}$ with large probability and otherwise, by $V_{2i}$.}

%\kimcomment{Justification of $Z_{2i}$ is illustrated on the above right after defining $w_i$}

%Besides,   we can also observe in $K_i$ 
%that $A_i + B_i$    is the weighted sum of $ S_{ji}^- / a_n \; (j = 1, \cdots, p_1)$ and %$S_{ji}^- / a_n - \Delta\; (j = p_1 + 1, \cdots, p)$, 
%\parkcomment{Isn't this $S_{ji}^- / a_n - \Delta $? \\ resolved: Yes, it should be $S_{ji}^- / a_n - \Delta $. }
%in which the sum of weights is $C_i \,+\, D_i$. This means that for some value 

%Therefore, from the intermediate value theorem,  $K_i \ge Z_{2i}$ always holds, regardless %of $E$. Consequently, we show \eqref{eqn:KiZi} on $E$.
%\parkcomment{Not clear yet. Provide some more explanation on how intermediate value theorem is used. }

\item In addition, one can rewrite $V_{1i}$ as
\begin{align*}
    V_{1i} & = \cfrac{h^2 Z_i^*}{1+h^2} + \cfrac{\cfrac{1}{\sqrt{2\pi}} \bigg (\cfrac{h^2Z_i^*}{1+h^2} + \cfrac{h (p_1-1) a_n \Delta T_i}{(1+h^2)^{3/2}} \bigg )-2p^{1/2+\epsilon}\bigg(\cfrac{h^2 Z_i^*}{1+h^2}+\cfrac{1}{h^2}\bigg)}
    {\cfrac{1}{\sqrt{2\pi}}+\cfrac{1}{\sqrt{2\pi}}\cfrac{h \left \{(p_1-1)T_i+(p-p_1)T_i^* \right \}}{(1+h^2)^{1/2}}-2p^{1/2+\epsilon}}.
\end{align*}

By plugging in the inequality $|Z_i| + a_n\Delta \le \sqrt{2\alpha(\log p + 1)}$ in $V_{1i}$ and the condition from $E_{5i}$ for $V_{2i}$, we obtain $W_{i} \le V_{i}$ for $W_i$ defined above.

Eventually, $W_i$ serves as the lower bound of $\widehat{\mu_{D,i}}^{EB}$ which are mutually independent. Therefore, by applying WLLN for triangular arrays on $W_i$, we show the second inequality of the lemma.

For large enough $p$ which satisfies 
\begin{equation}\label{eqn : large p condition}
    \log p \ge \sqrt{2\alpha(\log p + 1)},
\end{equation}

$$
|W_{1i}| \le  \cfrac{h^2}{1+h^2} \left\{\log p +\cfrac{(1 + 2\sqrt{2\pi}p^{1/2+\epsilon}) \log p  +   {(1 + h^2)/h^4}}
    {p^{1-\alpha}(1-p^{a-1}) - 2p^{1/2+\epsilon}}\right\}
$$
and 
$$|W_{2i}| \le \Delta + 2(\log p)^2 / a_n$$ 
holds on $E$ since $\displaystyle\max_{1 \le i \le n} |Z_i| < \log p$ and $1/h^2 = \log p$. Therefore, 
$$|W_i| \le 1 + 2(\log p)^2 / a_n.$$
In addition, from (\ref{eqn : large p condition}), one can obtain
\begin{align*}
   E\left[W_1 \right] &= E \left[ W_{1i} \,\bigg |\, |Z_i| + a_n\Delta \le \sqrt{2\alpha(\log p + 1)} \right] \times P \left(|Z_i|+a_n\Delta \le \sqrt{2\alpha(\log p + 1)}\right) \\
    &\quad + E \left [W_{2i} \,\bigg |\, |Z_i| + a_n\Delta > \sqrt{2\alpha(\log p + 1)} \right] \times P \left(|Z_i|+a_n\Delta > \sqrt{2\alpha(\log p + 1)} \right) \\
    & > \left\{ \frac{h^2}{1+h^2}a_n\Delta - \cfrac{1 + 2\sqrt{2\pi}p^{1/2 + \epsilon} + \log p}{p^{1 - \alpha} (1 - p^{a-1}) - 2p^{1/2 + \epsilon}} \right \}(1 - \gamma) + \left( \Delta - \cfrac{1}{a_nh^2} \log p \right) \gamma, 
\end{align*}

where we define $\gamma = P \left( |Z_i| + a_n\Delta > \sqrt{2\alpha(\log p + 1)} \right)$. Note that $$E \left[Z_i \, \bigg | \, |Z_i| < c\right] = 0.$$

By applying lemma \ref{lemma: maximum prob bound} and comparing the order of $p$ in each term, one can obtain that for $b$ which satisfy $b > \alpha + \epsilon - 1/2$ and $b > -\alpha$, 
$$
E \left [W_1 \right ] \ge \cfrac{h^2}{1+h^2}a_n\Delta \left\{ 1 + o(1) \right\}
$$ holds. Note that when $b > -1/4$, in which all $(a, b) \in \mathcal{R}_{NPEB}$ with $b<0$ are included, we can always pick such $\alpha, \epsilon>0$ as $\epsilon = b+ 1/4$, $\alpha = 1/4$. From the WLLN for triangular arrays, 

\begin{equation}
    \cfrac{\displaystyle\sum_{i=1}^{p_1} W_i - p_1E[W_1]}{1 + 2(\log p)^2 / a_n}  \overset{\mathrm{p}}{\to}  0.
\end{equation}
% \parkcomment{Provide more explanation why we need the above result. }
% \kimcomment{resolved: added more context about how to use obtained result.}
\end{itemize}

Therefore, one can write

$$
\sum_{i=1}^{p_1} W_i \ge \cfrac{h^2}{1+h^2}a_n p^{a+b} \left\{ 1 + o(1) \right \} + o_p\left\{ (\log p)^2 \right \}.
$$
\end{proof}

\section{Proof of Theorem 1}

%Before we provide the proof of Theorem 1,
%\noindent {\bf Proof of Theorem 1:} 

%Similar to the definition of $\mu_D^*$, 
We define $\bar{z}_D = \bar{z}_1-\bar{z}_2$ 
which has $\bar{z}_D \sim N\left(\mu_D^*, 
a_n^{-2} I_p \right)$
 where $a_n = \left(1/n_1+1/n_2 \right)^{-1/2}$. 
From (\ref{eqn: NPEB expression}), the estimator using EB method of mean difference $\mu_D^*$ is as follows:

\begin{equation}
    \widehat{\mu_{D,i}}^{EB} = \bar{z}_{D,i}+\frac{1}{h^2}
    \frac{\displaystyle\sum_{j=1}^p (\bar{z}_{D,j}-\bar{z}_{D,i}) \, \phi\left \{\cfrac{a_n(\bar{z}_{D,i}-\bar{z}_{D,j})}{h}\right \}}
    {\displaystyle\sum_{j=1}^p  \phi \left \{\cfrac{a_n(\bar{z}_{D,i}-\bar{z}_{D,j})}{h} \right \} }.
    \label{eqn:EB}
\end{equation}
Here, $\widehat{\mu_{D,i}}^{EB}, \bar{z}_{D,i}$ denotes the $i$-th component of the EB estimator of the mean vector  $\widehat{\mu_D}^{EB}$ and $\bar{z}_D$, respectively. $\phi(\cdot)$ is the probability distribution function of standard Gaussian distribution. 
%\hparkcomment{We need to define the $\phi(\cdot)$ when it is first introduced (Lemma 2.)}
%\kimcomment{resolved: added the description of $\phi$ and $\Phi$ in the proof of Lemma 2. }
According to \cite{brown2009nonparametric}  and \cite{greenshtein2009application}, 
we choose the bandwidth $h = 1/\sqrt{\log p}$ in \eqref{eqn:EB}. 

Since $\bar{z}_D \sim N(\mu_D^*, a_n^{-2}I_p)$, 
we have 
$\bar{z}_{D,i} \sim N(\mu_i^*, a_n^{-2})$ which are mutually independent. Thus, we can rewrite 
$\bar{z}_{D,i}$ as follows using independent  $Z_i \sim N(0,1)$ for $1\leq i \leq p$ : 
$$
\bar{z}_{D,i} = 
\begin{cases}
 \Delta +  Z_i / a_n, &  \mbox{if }1\leq i \leq p_1\\
Z_i / a_n, &  \mbox{if }p_1+1\leq i \leq p.
\end{cases}
$$
%Then 
%$S_1, \cdots, S_p \overset{\mathrm{iid}}{\sim} N(0, 1) $ holds.
For simplicity in expressions, we define the followings, %\btext
{aligned with lemma 4} : 
\begin{align}
   Z_{ij}^- = Z_i - Z_j, \quad  \quad &Z_i^*= Z_i + a_n \Delta,  \label{eqn:SijSi}\\
    T_i = \exp \left\{-\frac{Z_i^2}{2(1+h^2)} \right\}, \quad 
    &T_i^* =  \exp \left \{-\frac{Z_i^{*2}}{2(1+h^2)} \right \}. \label{eqn:TT}
\end{align}
With these notations,  $\widehat{\mu_{D,i}}^{EB}$ is represented as

\begin{equation}\label{eqn: mu hat}
\widehat{\mu_{D,i}}^{EB} = 
\begin{cases}
\Delta + \cfrac{Z_i}{a_n} + \cfrac{\displaystyle\sum_{j=1}^{p_1}\cfrac{Z_{ji}^{-}}{a_n} \, \phi (Z_{ij}^{-} / h) + 
\displaystyle\sum_{j=p_1+1}^{p} \cfrac{Z_{ji}^{-} - a_n \Delta}{a_n} \, \phi \left\{ (Z_{ij}^- + a_n\Delta )/h\right\}}
{h^2 \left [\displaystyle\sum_{j=1}^{p_1}\phi (Z_{ij}^- / h) + \displaystyle\sum_{j=p_1+1}^{p} \phi \left\{ (Z_{ij}^- + a_n\Delta )/h\right\} \right ]
}, & \mbox{if }1\leq i \leq p_1\\ \\

\cfrac{Z_i}{a_n} + 
\cfrac{\displaystyle\sum_{j=1}^{p_1}\cfrac{Z_{ji}^{-}+a_n\Delta}{a_n} \, \phi \left\{ (Z_{ij}^- - a_n\Delta )/h\right\} + 
\displaystyle\sum_{j=p_1+1}^{p} \cfrac{Z_{ji}^{-}}{a_n} \,\phi \left(Z_{ij}^- / h \right)}
{h^2 \left[ \displaystyle\sum_{j=1}^{p_1} \phi \left\{ (Z_{ij}^- - a_n\Delta )/h\right\} + \displaystyle\sum_{j=p_1+1}^{p} \phi \left(Z_{ij}^- / h \right)
\right ]}, &\mbox{if } p_1+1\leq i \leq p %\mbox{o.w.}
\end{cases}
\end{equation}

We now show that, 
for $(a, b) \in A\cup C\cup D$, %defined in section \ref{sec: Asymp}, 
 we have the divergence of $V$ {to $\infty$} in probability 
as in \eqref{eqn:perfect condition}. 
%where 
%\begin{equation}\label{eqn: mu mu product}
%\widehat{\mu_D^*}^T\mu_D^* = \Delta\sum_{i=1}^{p_1} \left(
%\Delta + \cfrac{S_i}{a_n} + \cfrac{\displaystyle\sum_{j=1}^{p_1}\cfrac{S_{ji}^{-}}{a_n} \, \phi (S_{ij}^{-} / h) + 
%\displaystyle\sum_{j=p_1+1}^{p} \cfrac{S_{ji}^{-} - a_n \Delta}{a_n} \, \phi \left\{ (S_{ij}^- + a_n\Delta )/h\right\}}
%{h^2 \left [\displaystyle\sum_{j=1}^{p_1}\phi (S_{ij}^- / h) + \displaystyle\sum_{j=p_1+1}^{p} \phi \left\{ (S_{ij}^- + a_n\Delta )/h\right\} \right ]
%}
%\right).
%\end{equation}

For this, we consider two cases depending on the sign of $b$: 
$(i)$ $b<0$ for $(a, b) \in A \cup D$ and $(ii)$ $b>0$ for  $(a, b) \in C$. 
%To obtain the sufficient condition for (\ref{eqn:perfect condition}), we analyzed the behavior of V according to the sign of $b$.
First, we show that among the region that satisfies $b<0$, the region $A\cup D$ is a subset of $\mathcal{R}_{N P E B}$.

%\subsection*{the proof for $A \cup D \subset \mathcal{R}_{NPEB}$} 
\subsection*{Claim I: $A \cup D \subset \mathcal{R}_{NPEB}$} 

In this case, we provide the proof of $\mathcal{R}_{NPEB} \supset A \cup D$, in which $(a, b)$ satisfies $a+2b>0$ and $-1/4 < b < 0$. 
For the proof of \eqref{eqn:perfect condition},  
we have the following roadmap consisting of two steps : 
\begin{enumerate}
    \item [ \it{ Step 1} :]  We present 
    a lower bound of the numerator in $V$ such as   
    $$\widehat{\mu_D^*}^T\mu_D^* \geq 
     \frac{h^2}{1+h^2} a_n p^{a+2 b}\left\{1+o_p(1)\right\}.$$
    \item [\it{Step 2} :] We then propose an upper bound of the denominator in $V$ as $$||\widehat{\mu_D^*}||_2=O_p\left\{p^{\max \left(2 a+2 b-1, \epsilon^{\prime}\right)}\right\},$$
    where $\epsilon'$ is an arbitrary positive constant.
\end{enumerate}

%To simplify  $\widehat{\mu_{D,i}}^{EB}$ for $i = 1, \cdots, p_1$, we define $K_i$ as follows:

%\parkcomment{Is $\widehat{\mu_{D,i}}^{EB}$ the same as $K_i$ for $1\leq i \leq p_1$? If it is, do we need to define $K_i$? We can use $\widehat{\mu_{D,i}}^{EB}$ for $1\leq i \leq p_1$.  }

%\begin{equation}\label{eqn: Ki definition}
%     K_i =
%   \Delta + \cfrac{S_i}{a_n} + \cfrac{\displaystyle\sum_{j=1}^{p_1}\cfrac{S_{ji}^{-}}{a_n} \, \phi (S_{ij}^{-} / h) + 
%\displaystyle\sum_{j=p_1+1}^{p} \cfrac{S_{ji}^{-} - a_n \Delta}{a_n} \, \phi \left\{ (S_{ij}^- + a_n\Delta )/h\right\}}
%{h^2 \left [\displaystyle\sum_{j=1}^{p_1}\phi (S_{ij}^- / h) + \displaystyle\sum_{j=p_1+1}^{p} \phi \left\{ (S_{ij}^- + a_n\Delta )/h\right\} \right ]
%}.
%\end{equation}

%Then $\widehat{\mu_D^*}^T\mu_D^* = \Delta \displaystyle\sum_{i=1}^{p_1}K_i$ holds. 

We present the proofs of $(i)$ Step 1 and $(ii)$ Step 2 as follows. 

%\begin{enumerate}
 %   \item [\textit{Proof of Step 1} :]

 \noindent
 \textit{$\bullet$ Proof of Step 1: }
From the expression of $\widehat{\mu_{D,i}}^{EB}$ in \eqref{eqn: mu hat},  one can observe that 
\begin{eqnarray}
\widehat{\mu_{D,i}}^{EB} 
%&=& 
%   \Delta + \cfrac{S_i}{a_n} + \cfrac{\displaystyle\sum_{j=1}^{p_1}\cfrac{S_{ji}^{-}}{a_n} \, \phi (S_{ij}^{-} / h) + 
%\displaystyle\sum_{j=p_1+1}^{p} \cfrac{S_{ji}^{-} - a_n \Delta}{a_n} \, \phi %\left\{ (S_{ij}^- + a_n\Delta )/h\right\}}
%{h^2 \left [\displaystyle\sum_{j=1}^{p_1}\phi (S_{ij}^- / h) + \displaystyle\sum_{j=p_1+1}^{p} \phi \left\{ (S_{ij}^- + a_n\Delta )/h\right\} \right ]
%}
     & \equiv & \Delta + \cfrac{Z_i}{a_n}+ \frac{1}{h^2} \left( \sum_{j=1}^{p_1} w_j\cfrac{Z_{ji}^-}{a_n} + \sum_{j=p_1+1}^p w_j  \cfrac{Z_{ji}^{-} - a_n \Delta}{a_n} \right) 
    %\cfrac{S_i}{a_n} + \cfrac{A_i + B_i} {h^2 \left (C_{i} + D_{i} \right )}
    \label{eqn: Ki weighted sum}
\end{eqnarray}
 for $1\leq i \leq p_1$,  where  $\sum_{j=1}^p w_j =1$ with 
\begin{eqnarray}
w_j =  \phi (Z_{ij}^{-}/h) \Bigg/
\left[\displaystyle\sum_{j=1}^{p_1}\phi (Z_{ij}^- / h) + \displaystyle\sum_{j=p_1+1}^{p} \phi \left\{ (Z_{ij}^- + a_n\Delta )/h\right\} \right]    
\end{eqnarray}
for $1\leq j \leq p_1$ and 
\begin{eqnarray}
w_j =  \phi \left\{ (Z_{ij}^- + a_n\Delta )/h\right\} \Bigg/
\left[\displaystyle\sum_{j=1}^{p_1}\phi (Z_{ij}^- / h) + \displaystyle\sum_{j=p_1+1}^{p} \phi \left\{ (Z_{ij}^- + a_n\Delta )/h\right\} \right]    
\end{eqnarray}
for $p_1+1 \leq j \leq p$. From (\ref{eqn: Ki weighted sum}), one can check that

\begin{equation}\label{eqn:Ki min bound}
\widehat{\mu_{D,i}}^{EB} \ge \Delta + \cfrac{Z_i}{a_n} + \cfrac{1}{a_n h^2}
    \min{\{(Z_{ji}^-)_{1 \le j \le p_1}, (Z_{ji}^- + a_n \Delta)_{p_1+1 \le j \le p}\}}
\end{equation}

always holds for $1\leq i \leq p_1$  
since the weighted mean is always greater than the minimum value.

%\parkcomment{Is $K_i$  defined only for $1\leq i \leq p_1$?  If it is, 
% is \eqref{eqn:Ki min bound} valid for $1\leq i\leq p_1$?  }
 
% \kimcomment{It holds for only $1\le i \le p_1$ and the form slightly changes when $p_1+1 \le i \le p+1$. The arguments in the min would be changed to $S_{ji}^- - a_n\Delta, S_{ji}^-$, respectively. Since we only need the values of  $\widehat{\mu_{D,i}}^{EB}$ for which $1\le i \le p_1$ to calculate $ \mu_D ^T \widehat{\mu_D^*}$, I stated for $ \widehat{\mu_{D,i}}^{EB}$ only for $1 \le i \le p_1$.}

Subsequently, we define $A_i, B_i, C_i, D_i$ as follows:

\begin{eqnarray*}
A_i &=& \sum_{j=1}^{p_1}\cfrac{Z_{ji}^{-}}{a_n} \, \phi (Z_{ij}^{-} / h) 
\equiv  \sum_{j=1}^{p_1} w_{1j} \cfrac{Z_{ji}^{-}}{a_n}, \\
B_i &=& \sum_{j=p_1+1}^{p}
\phi \left\{ (Z_{ij}^- + a_n\Delta )/h\right\} \cfrac{Z_{ji}^{-} - a_n \Delta}{a_n} \equiv \sum_{j=1}^{p_1} w_{2j} 
\cfrac{Z_{ji}^{-} - a_n \Delta}{a_n},\\
C_i &=& \sum_{j=1}^{p_1} \phi (Z_{ij}^- / h)
\equiv \sum_{j=1}^{p_1}w_{1j}, \\
D_i &=& \sum_{j=p_1+1}^{p} \phi \left\{ (Z_{ij}^- + a_n\Delta )/h\right\}
\equiv \sum_{j=p_1+1}^{p} w_{2j}.
\end{eqnarray*}

For $Z_i$ and $Z_j$ $(i \neq j)$,
one can check
\begin{eqnarray*}
E \left[ \phi (Z_{ji}^- / h)  \,\bigg |\, Z_i\right] &=& \cfrac{h}{\sqrt{2\pi}\sqrt{1+h^2}} T_i, \\ 
E \left[Z_{ij}^- \phi \left({Z_{ji}^-}/{h} \right) \,\bigg |\, Z_i\right] &=& - \cfrac{h^3}{\sqrt{2\pi}\sqrt{1+h^2}^3} Z_i T_i
\end{eqnarray*}

by direct calculation. A similar relation holds for $Z_i^*$ and $T_i^*$ as follows.
\begin{eqnarray*}
E \left[ \phi \left \{(Z_{ji}^- - a_n\Delta)/h \right\} \bigg | Z_i^*\right] &=& \cfrac{h}{\sqrt{2\pi}\sqrt{1+h^2}} T_i^*, \\ 
E \left [(Z_{ij}^- + a_n\Delta) \phi \left \{(Z_{ji}^- - a_n\Delta)/h \right\} \bigg | Z_i^*\right] &=& - \cfrac{h^3}{\sqrt{2\pi}\sqrt{1+h^2}^3} Z_i^* T_i^*.
\end{eqnarray*}

In addition, using the fact that $\phi(x) \in (0, 1/\sqrt{2\pi}]$, we obtain the following result : for $1\le i \le p_1$,  
$Z_{ij}^-$ are conditionally independent on $Z_i$, for all $j\neq i$. Therefore, for every $\epsilon > 0$, by Hoeffding's inequality, we have 
\begin{eqnarray}
P\left( \left|C_i - \frac{1}{\sqrt{2\pi}} -
    \frac{(p_1 - 1)h}{\sqrt{2 \pi}\sqrt{1+h^2}}T_i \right|>p_1^{1/2+\epsilon} \,\Bigg | \,Z_i \right) 
    \le 2\exp(-4\pi (p_1-1)^{2\epsilon})
    \label{eqn:Choeffiding}
\end{eqnarray}
since 
$$E \left[ C_i \,|\,Z_i \right] = E \left[ \sum_{j=1}^{p_1}\phi \left(Z_{ij}^{-}/h \right) \,\bigg |\,Z_i \right] =  \cfrac{1}{\sqrt{2\pi}} +
    \cfrac{(p_1 - 1)h}{\sqrt{2 \pi}\sqrt{1+h^2}}T_i.$$ 
By taking the expectation of $Z_i$ in \eqref{eqn:Choeffiding}, we have 
\begin{eqnarray}
P\left( \left|C_i - \frac{1}{\sqrt{2\pi}} -
    \frac{(p_1 - 1)h}{\sqrt{2 \pi}\sqrt{1+h^2}}T_i \right|>p_1^{1/2+\epsilon} \right ) \le 2\exp \left \{-4\pi (p_1-1)^{2\epsilon} \right \}.     
    \label{eqn:probboundC}
\end{eqnarray}
Similarly, we also have     
\begin{align}\label{eqn: prob bound CD}
    &P \left( \left| D_i -
    \frac{(p - p_1)h}{\sqrt{2\pi}\sqrt{1+h^2}} T_i^* \right| >(p-p_1)^{1/2+\epsilon} \right ) \le 2\exp \left \{-4\pi (p-p_1)^{2\epsilon} \right \}.
\end{align}
Also, from $|x\phi(x)| \le 1/\sqrt{2 \pi e}$, we obtain
\begin{align*}\label{eqn: prob bound AB}
     &P\left( \left|A_i +
    \cfrac{(p_1-1)h^3 Z_i}{\sqrt{2\pi}\sqrt{1+h^2}^3}T_i\right|>p_1^{1/2+\epsilon}  \right) \le 2\exp \left \{-\pi e (p_1-1)^{2\epsilon} \log p \right \}, \\
     &P\left(\left| B_i+ \cfrac{(p - p_1)h^3 Z_i^*}
     {\sqrt{2\pi}\sqrt{1+h^2}^3}T_i^*\right| > (p-p_1)^{1/2+\epsilon} \right) \le 2\exp \left \{-\pi e (p-p_1)^{2\epsilon} \log p \right \}.
\end{align*}
Now, we define the events $E_{1i}, \cdots, E_{5i}$ as follows:

\begin{eqnarray*}
    E_{1i} &=& \left\{\left|A_i +
    \cfrac{(p_1-1)h^3 Z_i}{\sqrt{2\pi}\sqrt{1+h^2}^3}T_i\right| < p_1^{1/2+\epsilon}  \right\}, \\
    E_{2i} &=& \left\{\left| B_i+ \cfrac{(p - p_1)h^3 Z_i^*}
     {\sqrt{2\pi}\sqrt{1+h^2}^3}T_i^*\right| < (p-p_1)^{1/2+\epsilon} \right\},  \\
    E_{3i} &=& \left\{ \left|C_i - \frac{1}{\sqrt{2\pi}} -
    \frac{(p_1 - 1)h}{\sqrt{2 \pi}\sqrt{1+h^2}}T_i \right| < p_1^{1/2+\epsilon} \right \}, \\
    E_{4i} &=& \left\{ \left| D_i -
    \frac{(p - p_1)h}{\sqrt{2\pi}\sqrt{1+h^2}}T_i^* \right| < (p-p_1)^{1/2+\epsilon} \right \}, \\
    E_{5i} &=& \left\{ \max_{1\le j \le p} |Z_j| + a_n\Delta < \log p \right\}.
\end{eqnarray*}
Let $E_i = \cap_{j=1}^5 E_{ji}$. Applying Bonferroni's inequality on $E_i^c$ and from Lemma \ref{lemma: max S_j prob}, we obtain
%\hparkcomment{
%It seems appropriate to put the lemmas used in developing the proof of Theorem 1 in front. It is awkward to use formula (\ref{eqn: max S_j prob}) that has not been introduced so far. BTW, if there are many lemmas to be introduced, it would be good to mention that (53) is a result introduced later.
%}
%\kimcomment{
%Resolved: I changed the order of the section, assigning section A for introducing lemmas and section B for the proof of theorem 1. However, after reversing the order, lemma 4 uses the equation (31) which is introduced later. If it is okay, we can keep the order.
%}
%%%%%%%%%%%%%%%%%%%%%%%%%%%%%%%%%%%%%%%%%%%%%%5 220924 14:52.
\begin{eqnarray*}
    P(E_i) &=& 1- P(E_i^c) = 1- P(\cup_{j=1}^{5} E_{ji}^c) \\
    &\ge& 1-\sum_{j=1}^{5} P(E_{ji}^c) \\
     &=&  1 - 8\exp \left \{-\pi e (p_1-1)^{2 \epsilon} \right \} - 1/p.
\end{eqnarray*}
Therefore, when we define $E = \cap_{i=1}^{p_1}E_i$, we have 
\begin{eqnarray}
    P(E) &=& 1- P(E^c) = 1-P(\cup_{i=1}^{p_1}E_i^c ) \nonumber\\
    &\ge& 1 -\sum_{i=1}^{p_1} \left \{ 1-P(E_i) \right \} \nonumber \\
    &\ge& 1 - 8p_1 \exp \left \{-\pi e (p_1-1)^{2 \epsilon} \right \} - 1/ p^{1-a} \label{eqn: prob bound for final event}
\end{eqnarray}

Note that for all $0 < a <1$, $\displaystyle\lim_{p \to \infty} P(E) = 1$. 
On the event $E$, we can set the bounds for $A_i, B_i, C_i, D_i$ to eventually build a lower bound of $ \widehat{\mu_{D,i}}^{EB}$. Since $\displaystyle\lim_{p \to \infty} P(E) = 1$, our aim is to define an appropriate lower bound, $W_i$, on $E$ which satisfies

\begin{eqnarray}
V = 
\frac{\widehat{\mu_D^*}^T\mu_D^*}{ 
||\widehat{\mu_D^*}||_2 } 
    = \frac{\Delta \sum_{i=1}^{p_1}  \widehat{\mu_{D,i}}^{EB} }{ 
||\widehat{\mu_D^*}||_2 }
     \geq  \frac{\Delta \sum_{i=1}^{p_1} W_i }{ 
||\widehat{\mu_D^*}||_2 }  \overset{\mathrm{p}}{\to} \infty.
\label{eqn:V}
\end{eqnarray}
for $(a, b) \in \mathcal{R}_{NPEB}$.
%To show 
%$$
%V = \frac{\widehat{\mu_D^*}^T\mu_D^*}{ 
%||\widehat{\mu_D^*}||_2 } \overset{\mathrm{p}}{\to} \infty
%$$
%for some $a$ and $b$, we want to obtain an appropriate lower bound of $K_i$, namely $Z_i$,  which satisfies

We directly apply the result in 
Lemma \ref{lemma:lowbound} and obtain  $\Delta\displaystyle\sum_{i=1}^{p_1}  \widehat{\mu_{D,i}}^{EB} \ge  
\Delta \displaystyle\sum_{i=1}^{p_1}W_i$ 
where $W_i$ is defined in Lemma 
\ref{lemma:lowbound} 
which leads to 
\begin{equation}
    \widehat{\mu_D^*}^T\mu_D^* = \Delta\sum_{i=1}^{p_1}  \widehat{\mu_{D,i}}^{EB} \ge 
    \Delta \displaystyle\sum_{i=1}^{p_1}W_i \ge \frac{h^2}{1+h^2}a_n p^{a+2b} \left \{1+o_p(1) \right \}.
\end{equation}
This inequality holds on $E$, and $E$ has the probability larger than $$1 - 8p_1\exp \left \{-\pi e (p_1 - 1)^{2\epsilon} \right \} - p^{-1+a},$$ which goes to 1 for all $0<a<1$ as $p \to \infty$. 
%\item [\textit{Proof of Step 2} :]

\noindent
 \textit{$\bullet$ Proof of Step 2: }
%\sout{The denominator in (\ref{eqn:perfect condition}) can be bounded using the result obtained in \cite{brown2009nonparametric} (equation 45).}
The denominator in (\ref{eqn:perfect condition}) can be bounded using the equation (45) in \cite{brown2009nonparametric}.
For $\nu = 1+h^2$, let

\begin{eqnarray*}
    \delta_i &:=& 
    \cfrac{\displaystyle\sum_{j=1}^{p_1}\Delta\phi \left\{(a_n\bar{z}_{D, i} - a_n\Delta)/\sqrt{\nu}\right\} }{\displaystyle\sum_{j=1}^{p_1}\phi \left\{(a_n\bar{z}_{D, i} - a_n\Delta)/\sqrt{\nu}\right\} + 
 \displaystyle\sum_{j=p_1+1}^{p} \phi \left( a_n\bar{z}_{D, i}/\sqrt{\nu}\right)}\\
 &=& 
 \begin{cases}
 \cfrac{p_1 \Delta \phi(Z_i / \sqrt{\nu})} {p_1\phi(Z_i / \sqrt{\nu}) +  (p-p_1)\phi \left\{(a_n\Delta + Z_i) / \sqrt{\nu}\right\}}, &  \mbox{if }1\leq i \leq p_1 \\
 \cfrac{p_1 \Delta\phi \left\{ (Z_i - a_n\Delta) / \sqrt{\nu} \right\}} {p_1\phi \left\{(Z_i - a_n\Delta) / \sqrt{\nu} \right \} +  (p-p_1)\phi \left(Z_i / \sqrt{\nu}\right)}, &  \mbox{if }p_1+1\leq i \leq p
 \end{cases}.
\end{eqnarray*}
Then the equation (45) in \cite{brown2009nonparametric} implies

\begin{equation}\label{eqn:brown and greenshtein}
 E \left[ \sum_{i=1}^{p}(\delta_i-\widehat{\mu_{D,i}}^{EB})^2 \right] = o(p^\epsilon)
\end{equation}
 for all $\epsilon > 0$. Using this, one obtain the probabilistic upper bound of $||\widehat{\mu}_D^{EB}||_2^2 $ from 

\begin{equation}\label{eqn:mu hat norm upper limit}
    ||\widehat{\mu}_D^{EB}||_2^2 \le 2\sum_{i=1}^p (\delta_i-\widehat{\mu_{D,i}}^{EB})^2 + 2\sum_{i=1}^p \delta_i^2.
\end{equation}

Since $Z_i$ are i.i.d., $E\left[\displaystyle\sum_{i=1}^p \delta_i^2 \right]$ is simplified as follows : 

\begin{align*}
     E \left[\sum_{i=1}^p \delta_i^2 \right] &= p_1 E \left[ 
      \cfrac{p_1 \Delta \phi(Z_i / \sqrt{\nu})} {p_1\phi(Z_i / \sqrt{\nu}) +  (p-p_1)\phi \left\{(a_n\Delta + Z_i) / \sqrt{\nu}\right\}}
     \right]^2\\
     &\quad+ (p-p_1) E \left[ 
     \cfrac{p_1 \Delta\phi \left\{ (Z_i - a_n\Delta) / \sqrt{\nu} \right\}} {p_1\phi \left\{(Z_i - a_n\Delta) / \sqrt{\nu} \right \} +  (p-p_1)\phi \left(Z_i / \sqrt{\nu}\right)}
     \right]^2\\
    & = p_1 E \left[\cfrac{p_1\Delta}{p_1 + (p-p_1)\exp \left\{-a_n\Delta(a_n\Delta + 2Z_i) \,/\, (2\nu) \right \}}\right]^2 \\
    &\quad + 
    (p-p_1) E \left[\cfrac{p_1 \Delta \exp \left\{-a_n\Delta(a_n\Delta - 2Z_i)\,/\,(2\nu) \right \}}
    {p_1  \exp \left [ -a_n\Delta(a_n\Delta - 2Z_i) \,/\, (2\nu) \right \} + (p-p_1)}\right]^2 \\
    & =: J_1 + J_2. \label{eqn:J1 J2 def}
\end{align*}

From $e^{-x} \ge 1-x$ for all $x$, under the event $E$, $J_1$ and $J_2$ satisfy

\begin{equation*}
    J_1 \le p_1\left \{\cfrac{p_1 \Delta}{p-(p-p_1)a_n\Delta(a_n\Delta + 2 \log p)/ 2}\right \}^2, \; J_2 \le (p-p_1) \left \{\cfrac{p_1 \Delta}{p_1 + (p-p_1)} \right \}^2,
\end{equation*}
which implies
    $E \left [ \sum_{i=1}^p \delta_i^2 \right ]\le p^{2a+2b-1}\left\{1+o(1)\right\}$. 
As a result, from (\ref{eqn:brown and greenshtein}) and (\ref{eqn:mu hat norm upper limit}), we have 
\begin{eqnarray}
||\widehat{\mu_D^{EB}}||_2^2 = O_p \left\{p^{\max(2a+2b-1, \epsilon')} \right\}   
\end{eqnarray}
for all $\epsilon' > 0$.
%\end{enumerate}

Therefore,  we prove the results in \textit{Step 1} and \textit{Step 2} leading to 
\begin{eqnarray*}
V &=& \frac{\widehat{\mu_D^*}^T\mu_D^*}{||\widehat{\mu_D^*}||_2} 
\geq \frac{h^2}{1+h^2} \frac{a_n p^{a+2 b} (1+o_p(1))  }{\sqrt{p^{\max(2a+2b-1, \epsilon')} O_p(1)} } \\
&=& \frac{a_n}{\log p} 
p^{a+2b - \max(a+b-1/2, \epsilon'/2 )} \frac{1+o_p(1)}{\sqrt{O_p(1)}} 
\overset{\mathrm{p}}{\to} 
\infty
\end{eqnarray*}

for $-1/4 < b < 0$ and  $a+2b > \epsilon' / 2$ since $a+2b - \max(a+b-1/2, \epsilon'/2 ) >0$  
for the given region of $(a,b)$. 

Subsequently, we show that the region $C$ is a subset of $\mathcal{R}_{N P E B}$ among the region that satisfies $b>0$.
\subsection*{Claim II: $C \subset \mathcal{R}_{NPEB}$}

%\btext
{Here, we show (\ref{eqn:perfect condition}) is satisfied for all $(a, b) \in C$ under the NPEB method.}
Since $\widehat{\mu_{D}^*}^{T}\mu_D^* = \Delta \sum_{i=1}^{p_1}  \widehat{\mu_{D,i}}^{EB}$ and $ \widehat{\mu_{D,i}}^{EB} \ge V_{2i}$ always holds,

\begin{align}
    \widehat{\mu_{D}^*}^{T}\mu_D^* &\ge \Delta \left( \sum_{i=1}^{p_1} \left [\Delta + \cfrac{1}{a_n}Z_i + \cfrac{1}{a_n h^2} \min{\{(Z_{ji}^-)_{1 \le j \le p_1}, (Z_{ji}^- + a_n \Delta)_{p_1+1 \le j \le p}\}} \right ] \right) \\
    &\ge \Delta \left[ \sum_{i=1}^{p_1} \left \{\Delta + \left(\cfrac{1}{a_n} - \cfrac{1}{a_n h^2} \right)Z_i - \cfrac{1}{a_n h^2} \min_{1 \le j \le p}{Z_j} \right \} \right].
\end{align}

Therefore, from Lemma \ref{lemma: max S_j prob} and the %\btext
{WLLN}, one can obtain $\widehat{\mu_{D}^*}^{T}\mu_D^* \ge p^{a+2b} \left \{1 + o_p(1) \right \}$. 

%\sout{In addition, by using $E \left [\exp(AS_i)\right ] = \exp(A^2/2)$  for $S_i \sim N(0, 1)$,}
In addition, by using $E \left [\exp(cZ_i)\right ] = \exp(c^2/2)$  for any constant $c>0$,
we have 
\begin{align*}
    E \sum_{i=1}^{p} \delta_i^2 &\le p_1 \Delta ^ 2 + (p - p_1) E \left [
    \cfrac{p_1 \Delta\exp \left\{ -a_n\Delta (a_n\Delta - 2Z_i) \,/\, (2\nu)\right\}
    }{p - p_1} \right ] ^2 \\
    & \le p^{a+2b} + \cfrac{p^{2a+2b}}{p - p_1}  = p^{a+2b} 
     + p^{2a+2b-1}(1+o(1))\\
     &= p^{a+2b}(1+o(1))  
\end{align*}
where the last equality is from $a<1$.  Therefore, from (\ref{eqn:brown and greenshtein}) and (\ref{eqn:mu hat norm upper limit}), we obtain 
\begin{equation}
    ||\widehat{\mu_D^*}||_2^2 = O_p(p^{a+2b}). 
\end{equation}
leading to 
$$
V = \cfrac{\widehat{\mu_D^*}^T\mu_D^*}{||\widehat{\mu_D^*}||_2} 
\ge  \cfrac{p^{a+2b} \left \{1+o_p(1) \right \}}{O_p \left \{p^{(a+2b)/2} \right \}} \overset{\mathrm{p}}{\to} \infty
$$
for all $a, b>0$. 

Therefore, our main Theorem \ref{thm:NPEB} is proved by Claims I and II.
\qed
\end{document}